\documentclass[3p, times]{elsarticle}
\usepackage{moreverb}
\usepackage{natbib}
\usepackage{float,lscape}
\usepackage{mathtools}
\usepackage[graphicx]{realboxes}
\usepackage{cuted}
\usepackage{flushend}
\usepackage{lipsum}
\newenvironment{proof}{\paragraph{Proof:}}{\hfill$\square$}
\usepackage{memhfixc}
\usepackage{lscape}
\usepackage{multirow}
\usepackage{graphicx}
\usepackage{lineno,hyperref}
\usepackage{amsmath}
\usepackage{latexsym}
\usepackage{amsfonts}
\usepackage{amssymb}
\usepackage{mathrsfs}
\usepackage[utf8]{inputenc}
\usepackage{bm}
\usepackage{graphicx}
\usepackage{hyperref}
\usepackage{subfigure}
\usepackage{pdfpages}
\newtheorem{theorem}{Theorem}[section]

\newtheorem{definition}{Definition}[section]
\newtheorem{lemma}[theorem]{Lemma}
\newtheorem{result}[theorem]{Result}

\newtheorem{corollary}[theorem]{Corollary}

\newcommand{\ba}{\begin{array}}
\newcommand{\ea}{\end{array}}

\newcommand{\bfl}{\begin{flushleft}}
\newcommand{\efl}{\end{flushleft}}
\newcommand{\bfr}{\begin{flushright}}
\newcommand{\efr}{\end{flushright}}
\newcommand{\bt}{\begin{theorem}}
\newcommand{\bd}{\begin{definition}}
\newcommand{\ed}{\end{definition}}
\newcommand{\et}{\end{theorem}}
\newcommand{\bl}{\begin{lemma}}
\newcommand{\el}{\end{lemma}}
\newcommand{\ee}{\end{exam}}
\newcommand{\bcor}{\begin{corollary}}
\newcommand{\ecor}{\end{corollary}}
\hypersetup{colorlinks,linkcolor={red},citecolor={blue},urlcolor={red}} 
\usepackage{multirow}
\usepackage{tabularx}
\usepackage{makecell}
\usepackage{rotating}
\usepackage{lscape}
\usepackage{float}
\usepackage{lipsum}
\usepackage{mathtools}
\usepackage{cuted}

\usepackage{longtable}
\begin{document}
\begin{frontmatter}
\title{Bayesian reliability acceptance sampling plans under adaptive simple step stress partial accelerated life test}
\author[label1]{Rathin Das \corref{cor1}}
\address[label1]{ 
Statistical Quality Control and Operations Research Unit, Indian Statistical Institute, 203, B. T. Road, Kolkata, PIN- 700108, India
}

\cortext[cor1]{Corresponding author}

\ead{rathindas65@gmail.com}

\author[label1]{Biswabrata Pradhan}

\begin{abstract}
    
{
In the traditional simple step-stress partial accelerated life test (SSSPALT), the items are put on normal operating conditions up to a certain time and after that the stress is increased to get the failure time information early. However, when the stress increases, an additional cost is incorporated that increases the cost of the life test. In this context, an adaptive SSSPALT is considered where the stress is increased after a certain time if the number of failures up to that point is less than a pre-specified number of failures.  We consider determination of Bayesian reliability acceptance sampling plans (BSP) through adaptive SSSALT conducted under Type I censoring. The BSP under adaptive SSSPALT is called BSPAA. The Bayes decision function and Bayes risk are obtained for the general loss function. Optimal BSPAAs are obtained for the quadratic loss function by minimizing Bayes risk. An algorithm is provided for computation of optimum BSPAA. Comparisons between the proposed BSPAA and the conventional BSP through non-accelerated life test (CBSP) and conventional BSP through SSSPALT (CBSPA) are carried out.}
\end{abstract}
\end{frontmatter}
\section*{Abbreviation \& Notation}
    \begin{longtable}{ll}
   SSSPALT& simple
step-stress partial accelerated life test\\
CEM & cumulative exposure model\\ 
RASP & reliability acceptance sampling plan\\
BSP& Bayesian reliability acceptance sampling plan\\
CBSP & conventional Bayesian sampling plan through non-accelerated life test\\
CBSPA& conventional Bayesian sampling plan through SSSPALT\\
BSPAA&  Bayesian sampling plan through adaptive SSSPALT\\
RRS& relative risk saving\\
$RRS_1$ &RRS of BSPAA over CBSP\\
$RRS_2$ &RRS of BSPAA over CBSPA\\
$n$ & sample size\\
$t_1,t_2$& time points\\
$m$& threshold value of number of failure up to $t_1$\\
$\boldsymbol{q}=(n,t_1,t_2,m)$& vector of sampling parameters\\
$E[D]$& expected number of failures under adaptive SSSALT\\
$E[\tau]$& expected time duration under adaptive SSSALT\\
$(n_B,t_{1B},t_{2B},m_B)$& optimal sampling parameter of BSPAA\\
$R_B$& Bayes risk of BSPAA\\
$E[D_B]$& expected number of failures for optimal sampling plan of BSPAA\\
$E[\tau_B]$& expected test duration for optimal sampling plan of BSPAA\\
$(n^*,t_1^*)$& optimal sampling parameter of CBSP\\
$R_1$& Bayes risk of CBSP\\
$E[D^*]$& expected number of failures for optimal sampling plan of CBSP\\
$E[\tau^*]$& expected test duration for optimal sampling plan of CBSP\\
$(n_A,t_{1A},t_{2A})$& optimal sampling parameter of CBSPA\\
$R_2$& Bayes risk of CBSPA\\
$E[D_A]$& expected number of failures for optimal sampling plan of CBSPA\\
$E[\tau_A]$& expected test duration for optimal sampling plan of CBSPA
\end{longtable}
\section{Introduction}

Reliability acceptance sampling plan (RASP) encompasses with acceptance or rejection of a lot based on reliability of the product under consideration. In RASP, \ a sample is taken from the lot for life testing and the decision of accepting or rejecting the lot is taken based on a suitable statistic obtained from lifetime data. Determination of RASP is an important task in reliability studies. There are various methods for selecting an optimal RASP. \ For example, producer's and consumer's risk point schemes, defense sampling schemes, Dodge and Roming's plan and decision-theoretic plans. Among these decision-theoretic approach is a more scientific and reasonable method from an economic point of view. This approach is widely used because it is decided upon by making the best choice based on some economic considerations, such as maximizing the return or minimizing the loss. 

Life tests are conducted to assess the reliability of the product. Normally, censored life tests are conducted because of time, cost and other resource constraints. Determination of RASP based on censored data is an important issue in reliability studies.
Type-I, type-II and hybrid censoring schemes are the most common types of censoring schemes under which life tests are conducted. In type-I censoring, the life test is terminated at a predetermined time $T_0$ and in type-II censoring scheme, the life test is terminated after a fixed number of failures $(r\leq n)$. Hybrid censoring scheme is a combination of type-I and type-II censoring schemes. There have been a number of works on the determination of RASPs under different censoring schemes by the Bayesian approach. For example, Yeh \cite{yeh1990optimal,yeh1994bayesian} obtained RASP under type II censoring and type I censoring, respectively. Yeh and Choy \cite{yeh1995bayesian} considered under random censoring and, Chen et al. \cite{chen2007bayesian}, Lin et al. \cite{lin2008exact} and 
Prajapati et al. \cite{prajapati2019new} obtained BSP under hybrid censoring. However, in these works, the optimal decision functions are not considered for finding BSPs. Lin \cite{lin2002bayesian} first introduced the Bayes decision function, which is the optimal decision function among all decision functions for the determination of BSP under type-I censoring. Later, for hybrid censored data, the Bayes decision approach was considered by Liang and Yang \cite{liang2013optimal} for the determination of BSP. 

In these works, the life test is conducted in normal operating conditions. However, in real-life scenarios, many products are highly reliable, the mean time to failure under normal operating conditions is very large. Also, under time constraints, it may not be possible to obtain enough lifetime information to take a decision with the use of conventional life testing experiments. In that situation, an accelerated life test (ALT) or a partial accelerated life test (PALT) are used to obtain enough lifetime information in a shorter period of time. 
 In ALTs, the items are put on a life test only on high-stress levels but in PALTs, the items are put on both normal and high-stress levels. The problem of designing BSP based on ALT or PALT with a censored sample has received less attention. Recently, Chen et al. \cite{chen2022designing, chen2023designing} and Prajapati and Kundu \cite{prajapati2023bayesian} have studied the designing BSP for a simple step-stress test based on type-II censored data. The step-stress test is one of the important life tests in ALT or PALT.  Under step-stress PALT (SSPALT), $n$  items are put on a life test at initial stress $s_0$ and then the stress is increased to $s_1<s_2<\cdots<s_k$ at pre-specified times $t_1<t_2<\cdots<t_k$ respectively. A simple step-stress PALT (SSSPALT) is a special case of SSPALT, when $k=1$. This means that after time $t_1$ the stress level $s_0$ is increased to $s_1$.  

In designing BSP under SSSPALT which has been studied by Chen et al. \cite{chen2022designing, chen2023designing}, it is seen that when stress increases, an additional cost is incorporated. In view of this, we consider an adaptive SSSPALT where changing the stress level after time $t_1$ depends on the number of failures up to time $t_1$. The adaptive test is described as follows. If the number of failures up to time $t_1$ is less than a pre-specified number $m$, then the stress $s_0$ changes to $s_1$. Otherwise, the test continues till $t_2$ with the stress $s_0$. If $m=0$, then the stress is unchanged after $t_1$ irrespective of the number of failures and the test is continued up to $t_2$ under initial stress $s_0$. In this case the life test becomes a conventional non-accelerated life test under type-I censoring. If $m=n$, then the stress is always changed to the stress level $s_1$ after $t_1$ irrespective of the number of failures till $t_1$ and the test is continued up to $t_2$ under the stress level $s_1$. In this case, the life test becomes a conventional accelerated life test under type-I censoring. The advantage of an adaptive test is that we can study conventional accelerated life test and non-accelerated life test together. Also, when an adaptive scenario occurs i.e., $0<m<n$, the cost of the life test is minimum than the non-accelerated and accelerated life tests.  Xiang et al. \cite{xiang2017designing} studied designing accelerated life tests under an adaptive scenario. 
However, there is no work on the determination of BSP under adaptive scenarios.  Also, to the best of our knowledge, there is no work on designing BSP based on SSSPALT under type-I censoring. Therefore, in this work, we obtain BSP through adaptive SSSALT and conventional SSSALT under type-I censoring.  

The paper is organized as follows. The framework of the BSPAA is described for the exponential distribution in Section \ref{model}. The Bayes decision function and Bayes risk are derived for the general loss function and prior distributions in Section \ref{bayesrisk}. The Bayes decision and Bayes risk are derived for the quadratic loss function in Section \ref{de}. We consider determination of optimal BSPAA in Section \ref{optim}. Optimum BSPAAs are computed under different scenarios and effect of the parameters on optimum BSPAA is studied in Section \ref{com}. This Section also describes a comparison among BSPAA and CBSP and CBSPA.  A data set is
analyzed to demonstrate the proposed model in Section \ref{nu}. The conclusion and future work are mentioned in Section \ref{con}. 

\section{Model and Assumptions}\label{model}

Suppose $n$ identical items are selected from the lot and put on life test in normal stress $s_0$ at time $t_0=0$. Let $t_1>0$ and $t_2>t_1$ denote the pre-fixed time points. Let $D_1$ be the number of failures that occur by time $t_1$  under the stress level $s_0$. If $D_1$ is less than a pre-assigned number $m ~(\leq n)$, the stress level $s_0$ is changed to a higher stress level $s_1$ and the life test continues up to time $t_2$. Otherwise, the stress level remains unchanged and the life test continues up to time $t_2$ with stress $s_0$. For simplicity, stress levels $s_0$ and $s_1$ are represented as follows.
\begin{align*}
  s=\begin{cases}
        s_0& \text{if } ~ 0<t<t_1\\
        s_{0}& \text{if } ~ t\geq t_1, ~D_1\geq m\\
       s_1&\text{if } ~ t\geq t_1, ~D_1<m.
    \end{cases}
\end{align*}
 We assume that the hazard rate at the stress level $s_i$ is $\lambda_{i}$, $i=0,1$. Let $Y$ denote the lifetime of an item with CDF $F$ and probability density function (PDF) $f$. From the cumulative exposure model (CEM) (for details, see Kundu \& Ganguly \cite{kundu2017analysis} and Nelson \cite{nelson2009accelerated}), the CDF of $Y$ is given by
\begin{align*}
    F(t)=\begin{dcases*}
        1-\exp(-\lambda_0 t)& \text{if } $t<t_1$\\
        [1-\exp(-\lambda_0 t)]^{1-\delta}[1-\exp[-\lambda_{0} t_1-\lambda_{1} (t-t_1)]]^\delta,&\text{if } $t\geq t_1$,
    \end{dcases*}
\end{align*}
where ${\delta}$ is the indicator function defined as
\begin{align}
    {\delta}=\begin{cases}
        1 & \text{if } D_1<m\\
        0 & \text{if } D_1\geq m.
    \end{cases}
\end{align}
The PDF of $Y$ is given by
\begin{align*}
    f(t)=\begin{dcases*}
        \lambda_0\exp(-\lambda_0 t)& \text{if } $t<t_1$\\
        [\lambda_0\exp(-\lambda_0 t)]^{1-\delta}[\lambda_1\exp[-\lambda_{0} t_1-\lambda_{1} (t-t_1)]]^\delta,&\text{if } $t\geq t_1$.
    \end{dcases*}
\end{align*}

Suppose $Y_1,Y_2,\ldots,Y_n$ be the lifetimes of $n$ items with PDF $f$ and $Z_1,Z_2,\ldots,Z_n$ be the order statistics corresponding to $Y_1,Y_2,\ldots,Y_n$. Let $D_2$ be the random variable denoting the number of failures in the interval $(t_1,t_2]$. The data corresponding to adaptive SSPALT is then given by $(\mathbf{Z},D_1,D_2)$, where $\mathbf{Z}=(Z_1,\ldots,Z_D)$ and $D=D_1+D_2$. The observed data is given by $(\mathbf{z},d_1,d_2)$.
Let $\lambda_0=\lambda$ and $\lambda_1=\phi\lambda$, where $\phi>1$ as the first stress level $s_0$ is less than the second stress level $s_1$. Let $\boldsymbol{\theta}=(\lambda,\phi)$ denote vector of parameters of the lifetime distribution. The likelihood function can be written as
\begin{align*}
    L( \boldsymbol{\theta}\ | \ (\boldsymbol{z},d_1,d_2))=
       \frac{n!}{(n-d)!}\phi^{\delta d_2}\lambda^d \exp\left[-\lambda\left(\sum_{j=1}^{d_1}z_j+(n-d_1)t_1+\phi^{\delta}\left(\sum_{j=d_1+1}^d(z_j-t_1)+(n-d)(t_2-t_1)\right)\right)\right].
\end{align*}
Let $w_1(\mathbf{z},d_1,d_2)=\sum_{j=1}^{d_1}z_j+(n-d_1)t_1$ and $w_2(\mathbf{z},d_1,d_2)=\sum_{j=d_1+1}^d(z_j-t_1)+(n-d)(t_2-t_1)$, then the likelihood becomes
\begin{align} \label{e1} L( \boldsymbol{\theta}\ | \ (w_1(\mathbf{z},d_1,d_2),w_2(\mathbf{z},d_1,d_2),d_1,d_2))\propto\lambda^d\phi^{d_2{\delta}}\exp\left[-\lambda \left(w_1(\mathbf{z},d_1,d_2)+\phi^{\delta} w_2(\mathbf{z},d_1,d_2)\right)\right].
\end{align}
Further, suppose $w_1 = $ $w_1(\mathbf{z},d_1,d_2)$ and $w_2 = w_2(\mathbf{z},d_1,d_2)$ and, denote $ \boldsymbol{x} = (w_1,w_2,d_1,d_2)$ as the observed data. The vector of the decision variable of the life testing plan under this setup is denoted by
$\boldsymbol{q} = (n, t_1, t_2, m)$. 

For developing the BSP, first, we assume the prior on parameters $\boldsymbol{\theta}=(\lambda,\phi)$ and then define the loss function. Let $p(\boldsymbol{\theta})$ be the prior distribution of $\boldsymbol{\theta}$. It is assumed that $\lambda$ and $\phi$ are independent and $p_1(.)$ and $p_2(.)$ be their PDFs, respectively. Therefore, $p(\boldsymbol{\theta})$ can be written as $p(\boldsymbol{\theta})=p_1(\lambda)p_2(\phi)$.
The posterior distribution of $\boldsymbol{\theta}$ given $\boldsymbol{x}$ is
\begin{align}\label{post}
    p(\boldsymbol{\theta} \ | \ \boldsymbol{x})=\frac{L(\boldsymbol{\theta}\ | \ \boldsymbol{x})p(\boldsymbol{\theta})}{p(\boldsymbol{x})},
\end{align}
where $p(\boldsymbol{x})=\int_{\boldsymbol{\theta}}L(\boldsymbol{\theta}\ | \ \boldsymbol{x})p(\boldsymbol{\theta}) ~d\boldsymbol{\theta}$.\\ Let $a(\boldsymbol{x}\ | \ \boldsymbol{q})$ denotes the action of acceptance sampling and it is defined by
\begin{align*}
    a(\boldsymbol{x}\ | \ \boldsymbol{q})=\begin{cases}
        1 & \text{if the lot is accepted for based on observed data $\boldsymbol{x}$}\\
        0 & \text{if the lot is rejected based on observed data $\boldsymbol{x}$}.
    \end{cases}
\end{align*}

It is assumed that the items are used by the consumer in normal operating conditions. If an item fail before the lifetime $L$, the loss due to failure is $h(\lambda)=CP(X<L)=C[1-\exp(-\lambda L)]$, where $C$ is the cost of accepting an item. If the hazard rate $\lambda$ increases, the chance of items failing before the lifetime $L$ increases. Therefore $h(\lambda)$ is a positive and increasing function with $\lambda$.
 In the paper, the order-restricted PALT is taken. Therefore, from the life-testing experiment, we obtain information about $\lambda$ and the accelerating factor $\phi$. Since $\lambda$ does not depend on $\phi$, $h(\lambda)$ does not depend on $\phi$. 
If the decision is rejection, the products are discarded or returned. Therefore, the rejection cost is fixed and denoted by $C_r$. 
We consider the loss functions for acceptance and rejection after the life testing as
\begin{align}\label{acc}
\mathcal{L}\left(a(\boldsymbol{x}\ | \ \boldsymbol{q})\ | \  \boldsymbol{\theta}\right)=
\begin{dcases*}
    h(\lambda)+nC_s-(n-d)v_s+\tau C_{t}+{\delta} (n-d_1)C_a & if $a(\boldsymbol{x}\ | \ \boldsymbol{q})\ | \  \boldsymbol{\theta})=1$\\
    C_r+nC_s-(n-d)v_s+\tau C_{t}+{\delta} (n-d_1)C_a,& if $a(\boldsymbol{x}\ | \ \boldsymbol{q})\ | \  \boldsymbol{\theta})=0$,
\end{dcases*} 
\end{align}
 where $C_s$ is the cost per item on life testing; $C_a$ is the additional cost per item for increasing the stress level from $s_0$ to $s_1$ if $D_1<m$; $v_s(<C_s)$ is the salvage value per item for the survived item after life testing; $C_t$ is the cost per unit time. \\

\section{Bayes Risk and Bayes decision function}\label{bayesrisk}
Here, we obtain the Bayes Risk using the loss function given in equation (\ref{acc}). Then we determine the life testing plan $\boldsymbol{q}_B$ and the optimal decision function $a_B(\boldsymbol{x}\ | \boldsymbol{q}_B)\equiv a_B$. We obtain $(\boldsymbol{q}_B,a_B)$ by minimizing the Bayes risk over all such sampling plans. 
\subsection{Bayes risk}
 The loss function in (\ref{acc}) $L(a(\boldsymbol{x}\ | \ \boldsymbol{q})\ | \ \boldsymbol{\theta})$ can be written as 
\begin{align*}
   \mathcal{L}\left(a(\boldsymbol{x}\ | \ \boldsymbol{q})\ | \  \boldsymbol{\theta}\right)=&\left[h(\lambda)+nC_s-(n-d)v_s+\tau C_{t}+{\delta} (n-d_1)C_a \right]a(\boldsymbol{x}\ | \ \boldsymbol{q})\\
   &+\left[C_r+nC_s-(n-d)v_s+\tau C_{t}+{\delta} (n-d_1)C_a\right]\left[1-a(\boldsymbol{x}\ | \ \boldsymbol{q})\right]\\
   =&a(\boldsymbol{x}\ | \ \boldsymbol{q})h(\lambda)+(1-a(\boldsymbol{x}\ | \ \boldsymbol{q}))C_r+nC_s+{\delta}(n-d_1)C_a-(n-d)v_s+\tau C_t.
\end{align*}
 The Bayes risk of a sampling plan $\boldsymbol{q}=(n,t_1,t_2,m)$ is obtained as 
\begin{align}\label{R}
R_B(\boldsymbol{q},a)&=E_{\boldsymbol{\theta}}E_{\boldsymbol{x}\ | \ \boldsymbol{\theta}}\left[\mathcal{L}\left(a(\boldsymbol{x}\ | \ \boldsymbol{q})\ | \  \boldsymbol{\theta}\right)\right]\nonumber\\
&=n(C_s-v_s)+C_a E_{\boldsymbol{\theta}}[(n-D_1)\ | \ D_1<m]+E_{\boldsymbol{\theta}}[D\ | \ \boldsymbol{\theta}]+C_t E_{\boldsymbol{\theta}}[\tau\ | \ \boldsymbol{\theta}]+R_1(\boldsymbol{q},a)\nonumber\\
&=n(C_s-v_s)+C_a n_{as}+E[D]+C_t E[\tau]+R_1(\boldsymbol{q},a)
\end{align}
where
\begin{align}\label{R1}
R_1(\boldsymbol{q},a)&=E_{\boldsymbol{\theta}}E_{\boldsymbol{x}\ | \ \boldsymbol{\theta}} [a(\boldsymbol{x}\ | \ \boldsymbol{q})h(\lambda)+(1-a(\boldsymbol{x}\ | \ \boldsymbol{q}))C_r],
\end{align}
$n_{as}= E_{\boldsymbol{\theta}}[(n-D_1)\ | \ D_1<m]$, $E[D]=E_{\boldsymbol{\theta}}[D\ | \ \boldsymbol{\theta}]$ and $E[\tau]=E_{\boldsymbol{\theta}}[\tau \ | \ \boldsymbol{\theta}]$.
The expressions of $E[D]$, $E[\tau]$ and $n_{as}$ are provided in \ref{appB}.

\subsection{Bayes decision function \texorpdfstring{${a}_B(.\ | \ {\boldsymbol{q}})$}{}}\label{bayesa}
 We obtain the Bayes decision function that minimizes the Bayes risk $R_B(\boldsymbol{q},a)$ among the class of all decision functions. The Bayes risk $R_B(\boldsymbol{q},a)$ in equation (\ref{R}) consists of the decision function $a(\cdot\ | \ \boldsymbol{q})$ only in the term $R_1(\boldsymbol{q},a)$. Note that
$R_1(\boldsymbol{q},a)$ in equation (\ref{R1}) can be written as
\begin{align*}
   R_1(\boldsymbol{q},a)  =E_{\boldsymbol{\theta}}[h(\lambda)]+E_{\boldsymbol{\theta}}E_{\boldsymbol{x}\ | \ \boldsymbol{\theta}} [(1-a(\boldsymbol{x}\ | \ \boldsymbol{q}))(C_r-h(\lambda))]\\
   =E_{\boldsymbol{\theta}}[h(\lambda)]+E_{\boldsymbol{x}}E_{\boldsymbol{\theta}\ | \ \boldsymbol{x}} [(1-a(\boldsymbol{x}\ | \ \boldsymbol{q}))(C_r-h(\lambda))].
\end{align*}
we have
\small\begin{align*}
&E_{\boldsymbol{x}}E_{\boldsymbol{\theta}\ | \ \boldsymbol{x}} [(1-a(\boldsymbol{x}\ | \ \boldsymbol{q}))(C_r-h(\lambda))]\\
=&\sum_{d_1=0}^n\sum_{d_2=0}^{n-d_1}\int_{w_1}\int_{w_2}[1-a(w_1,w_2,d_1,d_2\ | \ \boldsymbol{q})]E_{\boldsymbol{\theta}\ | \ w_1,w_2,d_1,d_2}[C_r-h(\lambda)]~g_{(W_1,W_2,D_1,D_2)}(w_1,w_2,d_1,d_2\ | \ \boldsymbol{\theta})~dw_1~dw_2\\
=&\sum_{d_1=0}^n\sum_{d_2=0}^{n-d_1}\int_{w_1}\int_{w_2}[1-a(w_1,w_2,d_1,d_2\ | \ \boldsymbol{q})]\left\{C_r-\int_{\boldsymbol{\theta}}h(\lambda)~p( \boldsymbol{\theta}\ | \  w_1,w_2,d_1,d_2) ~d\boldsymbol{\theta}\right\}~g_{(W_1,W_2,D_1,D_2)}(w_1,w_2,d_1,d_2\ | \ \boldsymbol{\theta})~dw_1~dw_2\\
=&\sum_{d_1=0}^n\sum_{d_2=0}^{n-d_1}\int_{w_1}\int_{w_2}[1-a(w_1,w_2,d_1,d_2\ | \ \boldsymbol{q})][C_r-\varphi(w_1,w_2,d_1,d_2)]~g_{(W_1,W_2,D_1,D_2)}(w_1,w_2,d_1,d_2\ | \ \boldsymbol{\theta})~dw_1~dw_2,
\end{align*}
\normalsize
where $g_{(W_1,W_2,D_1,D_2)}$ is the joint distribution of $(W_1,W_2,D_1,D_2)$ and 
\begin{align*}
\varphi(w_1,w_2,d_1,d_2)=\int_{\boldsymbol{\theta}}h(\lambda)~p( \boldsymbol{\theta}\ | \  w_1,w_2,d_1,d_2) ~d\boldsymbol{\theta}.
\end{align*}
To obtain Bayes decision, we need to minimize $R_1(\boldsymbol{q},a)$ with respect to $a$ which is equivalent to minimization of $E_{\boldsymbol{\theta}}E_{\boldsymbol{x}\ | \ \boldsymbol{\theta}} [(1-a(\boldsymbol{x}\ | \ \boldsymbol{q}))(C_r-h(\lambda))]$ with respect to $a$.
Now we consider two cases:\\
\textbf{Case 1: } $\varphi(w_1,w_2,d_1,d_2)\leq C_r$:\\
If $a(w_1,w_2,d_1,d_2\ | \ \boldsymbol{q})=1$, then $ E_{(W_1,W_2,D_1,D_2)}E_{ \boldsymbol{\theta}\ | \ (w_1,w_2,d_1,d_2)} [(1-a(\boldsymbol{x}\ | \ \boldsymbol{q}))(C_r-h(\lambda))]=0$ and if $a(w_1,w_2,d_1,d_2\ | \ \boldsymbol{q})=0$, then $ E_{(W_1,W_2,D_1,D_2)}E_{ \boldsymbol{\theta}\ | \ (w_1,w_2,d_1,d_2)} [(1-a(\boldsymbol{x}\ | \ \boldsymbol{q}))(C_r-h(\lambda))]\geq0$\\
\textbf{Case 2: }$\varphi(w_1,w_2,d_1,d_2)> C_r$:\\
If $a(w_1,w_2,d_1,d_2\ | \ \boldsymbol{q})=1$, then $ E_{(W_1,W_2,D_1,D_2)}E_{ \boldsymbol{\theta}\ | \ (w_1,w_2,d_1,d_2)} [(1-a(\boldsymbol{x}\ | \ \boldsymbol{q}))(C_r-h(\lambda))]=0$ and if $a(w_1,w_2,d_1,d_2\ | \ \boldsymbol{q})=0$, then $ E_{(W_1,W_2,D_1,D_2)}E_{ \boldsymbol{\theta}\ | \ (w_1,w_2,d_1,d_2)} [(1-a(\boldsymbol{x}\ | \ \boldsymbol{q}))(C_r-h(\lambda))]\leq0$\\
Therefore, for each fixed $\boldsymbol{q}$, if we take $ C_r- \varphi(w_1,w_2,d_1,d_2)\geq 0$ when $a(w_1,w_2,d_1,d_2\ | \ \boldsymbol{q})=1$ and $C_r- \varphi(w_1,w_2,d_1,d_2)$ $\leq 0$ when $a(w_1,w_2,d_1,d_2\ | \ \boldsymbol{q})=0$,
then $E_{\boldsymbol{\theta}}E_{\boldsymbol{x}\ | \ \boldsymbol{\theta}} [(1-a(\boldsymbol{x}\ | \ \boldsymbol{q}))(C_r-h(\lambda))]$ is minimized with respect to $a(\cdot\ | \ \boldsymbol{q})$.
So for fixed $\boldsymbol{q}$, the Bayes decision function $a_B(\boldsymbol{x}\ | \ \boldsymbol{q})$ is given by,
\begin{align*}
    a_B(w_1,w_2,d_1,d_2\ | \ \boldsymbol{q})=\begin{cases}
        1 &\text{ if } C_r-\varphi(w_1,w_2,d_1,d_2)\geq 0\\
        0 &\text{ otherwise}.
    \end{cases}
\end{align*}
Next, we provide an alternative form of the Bayes decision function that is useful to obtain a simplified form of the term $R_1(\boldsymbol{q},a_B)$ in the Bayes risk. 
\subsection{Alternative form of Bayes decision function}\label{de1}
For developing an alternative form of the Bayes decision function, we consider the following theorem.
\begin{theorem}\label{the} If $h(\lambda)$  is increasing in $\lambda$, then the posterior expectation of $h(\lambda)$ which is given by $\varphi(w_1,w_2,d_1,d_2)$, satisfies the following monotonicity properties:
    \begin{enumerate}[(i)]
\item For fixed $(w_2,d_1,d_2)$, $\varphi(w_1,w_2,d_1,d_2)$ is decreasing  in $w_1$.
\item For fixed $(w_1,d_1,d_2)$, $\varphi(w_1,w_2,d_1,d_2)$ is decreasing in $w_2$.
    \end{enumerate}
\end{theorem}
\begin{proof}
The proof is given in the \ref{appA}.
\end{proof}

\noindent For fixed $(w_1=0,d_1,d_2)$, $\varphi(0,w_2,d_1,d_2)$ is a decreasing function in $w_2$. If $\varphi(0,0,d_1,d_2)>C_r$, then there exists a unique point $c_1(d_1,d_2)$ such that
\begin{align*}
        \varphi(0,w_2,d_1,d_2)>C_r, &~~~~~ \text{for } w_2<c_1(d_1,d_2)\\
\varphi(0,w_2,d_1,d_2)<C_r, &~~~~~ \text{for } w_2>c_1(d_1,d_2).
\end{align*}
 If $\varphi(0,0,d_1,d_2)<C_r$, then $\varphi(0,w_2,d_1,d_2)<C_r$ for all $w_2>0$. In that scenario, $c_1(d_1,d_2)$ can be taken as $0$.
Now, for fixed $(d_1,d_2,w_2=c_1(d_1,d_2))$, $\varphi(d_1,d_2,w_1,c_1(d_1,d_2))$ is a decreasing function in $w_1$. Therefore,
\begin{align*}
\varphi(w_1,w_2,d_1,d_2)<C_r,&~~~~~\text{for }w_1>0 ~\text{and }  w_2>c_1(d_1,d_2).
\end{align*}
When $0<w_2<c_1(d_1,d_2)$, we get $\varphi(0,w_2,d_1,d_2)>C_r$. Since for fixed $(d_1,d_2,w_2)$, $\varphi(w_1,w_2,d_1,d_2)$ is decreasing function in $w_1$. Therefore, when $0<w_2<c_1(d_1,d_2)$ there exists a unique point $c_2(d_1,d_2,w_2)$ such that 
\begin{align*}
        \varphi(w_1,w_2,d_1,d_2)<C_r, &~~~~~ \text{for } w_1>c_2(d_1,d_2,w_2)~\text{and } 0<w_2<c_1(d_1,d_2) \\
        \varphi(w_1,w_2,d_1,d_2)>C_r, &~~~~\text{for } 0<w_1<c_2(d_1,d_2,w_2)~\text{and } 0<w_2<c_1(d_1,d_2).
\end{align*} 
Note that the upper bounds of $w_1$ and $w_2$ are $nt_1$ and $(n-d_1)(t_2-t_1)$, respectively. Define $c'_1(d_1,d_2)=\min\{c_1(d_1,d_2),nt_1\}$ and $c'_2(d_1,d_2,w_2)=\min\{c_2(d_1,d_2,w_2),(n-d_1)(t_2-t_1)\}\}$.
Therefore, the Bayes decision function can be written as 
\begin{align*}
    a_B(w_1,w_2,d_1,d_2\ | \ \boldsymbol{q})=\begin{dcases*}
        1 & for $w_2>c'_1(d_1,d_2)$ or $0<w_2<c'_1(d_1,d_2)$ and $w_1>c'_2(d_1,d_2,w_2)$\\
        0& for $0<w_2<c'_1(d_1,d_2)$ and $0<w_1<c'_2(d_1,d_2,w_2)$.
    \end{dcases*}
\end{align*}
\subsection{Alternative form of \texorpdfstring{$R_1(\boldsymbol{q},a_B)$}{}}
Now, Using the Bayes decision function, $R_1(\boldsymbol{q},a_B)$ in equation (\ref{R1}) can be written as
\allowdisplaybreaks\begin{align*}
R_1(\boldsymbol{q},a_B)=&E_{\boldsymbol{\theta}}[h(\lambda)]+\sum_{d_1=0}^n\sum_{d_2=0}^{n-d_1}\int_{\boldsymbol{\theta}}[C_r-h(\lambda)]P(W_1<c'_2(d_1,d_2,w_2), W_2<c'_1(d_1,d_2),D_1=d_1,D_2=d_2\ | \ \boldsymbol{\theta})~p(\boldsymbol{\theta})~d\boldsymbol{\theta}\\
=&E_{\boldsymbol{\theta}}[h(\lambda)]+\sum_{d_1=0}^n\sum_{d_2=0}^{n-d_1}\int\limits_{w_2=0}^{c'_1(d_1,d_2)}\int\limits_{w_1=0}^{c'_2(d_1,d_2,w_2)}\int_{\boldsymbol{\theta}}[C_r-h(\lambda)]f_{(W_1,W_2,D_1,D_2)}(w_1,w_2,d_1,d_2\ | \ \boldsymbol{\theta})p(\boldsymbol{\theta})~dw_1~dw_2~d\boldsymbol{\theta}\\
=&E_{\boldsymbol{\theta}}[h(\lambda)]+\sum_{d_1=0}^n\sum_{d_2=0}^{n-d_1} H(d_1,d_2),
\end{align*}
where 
\begin{align*}
    H(d_1,d_2)=\int\limits_{w_2=0}^{c'_1(d_1,d_2)}\int\limits_{w_1=0}^{c'_2(d_1,d_2,w_2)}\int_{\boldsymbol{\theta}}[C_r-h(\lambda)]f_{(W_1,W_2,D_1,D_2)}(w_1,w_2,d_1,d_2\ | \ \boldsymbol{\theta})p(\boldsymbol{\theta})~dw_1~dw_2~d\boldsymbol{\theta}.\\  
\end{align*}
The joint distribution function of $(W_1,W_2,D_1,D_2)$ is given in  \ref{appb2}.
\section{Bayes decision function and Bayes risk for quadratic loss function}\label{de}
Here, the loss function is taken as a quadratic loss function
$h(\lambda)=a_0+a_1\lambda+a_2\lambda^2$ as considered by Yeh \cite{yeh1994bayesian}.  The prior distributions of $\lambda$ and $\phi$ are taken as $\lambda$ follows a gamma distribution  with PDF
\begin{align}
p_1(\lambda)=\frac{\beta^\alpha}{\Gamma(\alpha)}\lambda^{\alpha-1}\exp(-\beta\lambda),~~~\lambda>0,~\alpha>0,~ \beta>0.
\end{align}
and $\phi$ follows uniform distribution with PDF
\begin{align}
    p_2(\phi)=\frac{1}{l-1}, ~~~~~1<\phi<l.
\end{align}
 Now, for each $\boldsymbol{q}=(n,t_1,t_2,m)$, we derive the Bayes decision function $a_B( . \ | \ \boldsymbol{q})$. First, we compute the joint posterior distribution of $(\lambda,\phi)$ for given $\boldsymbol{x}$. Using the equation (\ref{post}), we get
 \begin{align}
     p(\lambda,\phi \ | \ \boldsymbol{x})=\frac{\phi^{{\delta} d_2}\lambda^{d+\alpha-1}\exp[-\lambda(w_1+\phi^{\delta} w_2+\beta)]}{\int_0^\infty\int_1^l \phi^{{\delta} d_2}\lambda^{d+\alpha-1}\exp[-\lambda(w_1+\phi^{\delta} w_2+\beta)]~d\lambda~ d\phi}.
 \end{align}
To obtain Bayes decision, we need the following result:
\begin{result}\label{re1}
If $p$ is a non-negative real number, then we get
    \begin{align*}
     \int_1^l \int_0^\infty \phi^{{\delta} d_2} \lambda^p \exp(-\lambda (w_1+\phi ^{\delta} w_2))\lambda^{d_1+d_2}\lambda^{\alpha-1}\exp(-\beta\lambda) ~d\lambda~d\phi&=\int_1^l\phi^{{\delta} d_2}\frac{\Gamma(p+d_1+d_2+\alpha)}{(w_1+\phi ^{\delta} w_2+\beta)^{p+d_1+d_2+\alpha}}~d\phi\\
     &=H_1(w_1,w_2,d_1,d_2,p), ~\text{say}.
    \end{align*}
 .
\end{result}
If ${\delta}=0$, 
\begin{align*}
H_1(w_1,w_2,d_1,d_2)=\frac{\Gamma(p+d_1+d_2+\alpha)(l-1)}{(w_1+ w_2+\beta)^{p+d_1+d_2+\alpha}}
\end{align*}
and if ${\delta}=1$,
\begin{align*}
H_1(w_1,w_2,d_1,d_2)=&\int_1^l\phi^{d_2}\frac{\Gamma(p+d+\alpha)}{(w_1+\phi w_2+\beta)^{p+d+\alpha}}~d\phi\\
=&\int_1^l\phi^{d_2}\frac{\Gamma(p+d+\alpha)}{(w_1+\beta)^{p+d+\alpha}(1+\phi w_2/(w_1+\beta))^{p+d+\alpha}}~d\phi\\
=&\int_{w_2/(w_1+\beta)}^{w_2l/(w_1+\beta)}\frac{\Gamma(p+d+\alpha)}{(w_1+\beta)^{p+d_1+\alpha+1}w_2^{d_2+1}}\frac{z^{d_2}}{(1+z)^{p+d+\alpha}}~d\phi\\
&=\frac{\Gamma(d_2+1)\Gamma(p+d_1+\alpha-1)}{(w_1+\beta)^{p+d_1+\alpha+1}w_2^{d_2+1}}[I_{\zeta_1}(d_2+1,p+d_1+\alpha-1)-I_{\zeta_2}(d_2+1,p+d_1+\alpha-1)],
\end{align*}
where $\zeta_i=\eta_i/(1+\eta_i)$, for $i=1,2$, $\eta_1=w_2l/(w_1+\beta)$, $\eta_2=w_2/(w_1+\beta)$ and  $I_\eta(m,b)={B_\eta(m,b)}/{B(m,b)}$ is the cdf of beta function given by $I_\eta(m,b)$, where $B_\eta(m,b)$ is the incomplete beta function given by \begin{align*}
    B_\eta(m,b)=\int_0^\eta x^{m-1}(1-x)^{b-1} dx.
\end{align*}
Using Result \ref{re1}, $\varphi(\boldsymbol{x})$ can be written as
\begin{align}
 \varphi(\boldsymbol{x})=  a_0+a_1\frac{H_1(w_1,w_2,d_1,d_2,1)}{H_1(w_1,w_2,d_1,d_2,0)}+a_2\frac{H_1(w_1,w_2,d_1,d_2,2)}{H_1(w_1,w_2,d_1,d_2,0)}. 
\end{align}
Therefore, the Bayes decision can be obtained as follows
\begin{align*}
    a_B(w_1,w_2,d_1,d_2\ | \ \boldsymbol{q})=\begin{dcases*}
        1& if $a_0+a_1\frac{H_1(w_1,w_2,d_1,d_2,1)}{H_1(w_1,w_2,d_1,d_2,0)}+a_2\frac{H_1(w_1,w_2,d_1,d_2,2)}{H_1(w_1,w_2,d_1,d_2,0)} \leq C_r$\\
        0 &otherwise.
    \end{dcases*}
\end{align*}

Note that the prior expectation of the loss $h(\lambda)$ is
\begin{align*}
    E_{(\lambda,\phi)}[h(\lambda)]=\varphi(0,0,0,0)=\int_1^l\int_0^\infty h(\lambda)p_1(\lambda)p_2(\phi)~d\lambda~d\phi=a_0+\frac{a_1\alpha}{\beta}+a_2\frac{a_2\alpha(\alpha+1)}{\beta^2}.
\end{align*}
For no sampling case, the Bayes decision function is given by
\begin{align*}
    a_B(0,0,0,0 \ | \ 0,0,0,0)=\begin{dcases*}
        1 & if $a_0+\frac{a_1\alpha}{\beta}+a_2\frac{a_2\alpha(\alpha+1)}{\beta^2}\leq C_r$\\
        0& otherwise.
    \end{dcases*}
\end{align*}
Using the decision function $a(\cdot \ | \ \boldsymbol{q})$, the explicit form of Bayes risk is given in the  \ref{appb1}.

\section{Optimal BSPAA}\label{optim}
Here we find optimal BSPAA. The optimal BSPAA is obtained using Algorithm A.

\textbf{Algorithm A}
\begin{enumerate}
   \item Choose a sufficiently large value of $n$ say $n_0$.
    \item For each $\boldsymbol{q}=(n, t_{1}, t_{2}, m)$, derive the Bayes decision function $a_B(\cdot\ | \ \boldsymbol{q})$ to minimize $R_B(\boldsymbol{q})$ among all class of decision functions  $a(\cdot\ | \ \boldsymbol{q})$. The derivation of the Bayes decision function is discussed in Sections \ref{bayesa} and \ref{de1}.
    \item For each pair of values of $(n,m)$, minimize the Bayes risk $R_B(\boldsymbol{q},a)$ with respect to $(t_1,t_2)$. Since $t_2>t_1$, we take $t_2=t_1+h$ and minimize the equation (\ref{R}) with respect to $(t_1,h)$, where $t_1>0$ and $h>0$. Let $t_{1B}(n,m)$ and $h_B(n,m)$ be the optimal values of $t_1$ and $h$ respectively, then the optimal value of $t_2$ is $t_{2B}(n,m)=t_{1B}(n,m)+h_B(n,m)$ and
    \begin{flalign*}
       &R_B(n,t_{1B}(n,m),t_{2B}(n,m),m,a_B(\cdot \ | \ (n,t_{1B}(n,m),t_{2B}(n,m),m)))&&\\
       & ~~~~ =\min_{t_1>0,t_2>t_1} R_B(n,t_1,t_2,m,a_B(\cdot \ | \ (n,t_{1},t_{2},m))). &&
    \end{flalign*}
    \item For each value of $n$, find an integer $m_B(n)$, $0\leq m_B(n)\leq n$ which minimizes the equation (\ref{R}) with respect to $m$, that is,
    \begin{flalign*}
       & R_B(n,t_{1B}(n,m_B(n)),t_{2B}(n,m_B(n)),m_B(n),a_B(\cdot \ | \ (n,t_{1B}(n,m_B(n)),t_{2B}(n,m_B(n)),m_B(n))))&&\\
        &~~~~=\min_{0\leq m\leq n} R_B(n,t_{1B}(n,m),t_{2B}(n,m),a_B(\cdot \ | \ (n,t_{1B}(n,m),t_{2B}(n,m),m))).&&
    \end{flalign*}
    \item Finally, we find $n_B$, $0\leq n\leq n_0$ which minimizes the equation (\ref{R}), that is,
    \begin{flalign*}
       & R_B(n_B,t_{1B}(n_B,m_B(n_B)),t_{2B}(n_B,m_B(n_B)),m_B(n_B),a(\cdot \ | \ (n_B,t_{1B}(n_B,m_B(n_B)),t_{2B}(n_B,m_B(n_B)),m_B(n_B))&&\\
       &~~~~ =\min_{0<n\leq n_0} R_B(n,t_{1B}(n,m_B(n)),t_{2B}(n,m_B(n)),m_B(n),a(\cdot \ | \ (n,t_{1B}(n,m_B(n)),t_{2B}(n,m_B(n)),m_B(n)))).&&
    \end{flalign*}
    \end{enumerate}
Now, we write $a(\cdot \ | \ (n_B,t_{1B}(n_B,m_B(n_B),t_{2B}(n_B,m_B(n_B)),m_B)))=a_B$, $t_{1B}(n_B,m_B(n_B))=t_{1B}$, $t_{2B}(n_B,m_B(n_B))=t_{2B}$ and $m_B(n_B)=m_B$.  Therefore, $(\boldsymbol{q}_B,a_B)=(n_B,t_{1B},t_{2B},m_{B},a_B)$ is the optimal sampling plan. 
    \begin{theorem}
        The sampling plan $(n_B,t_{1B},t_{2B},m_B,a_B)$ is the optimal BSPAA
    \end{theorem}
    \begin{proof}
        It is enough to prove that for any sampling plan $(n,t_1,t_2,m,a)$, the following inequality holds:
        \begin{align*}
           R_B(n,t_{1},t_{2},m,a)\geq R_B(n_B,t_{1B},t_{2B},m_B,a_B).
        \end{align*}
        Now, 
        \begin{align*}
            R_B(n,t_1,t_2,m,a)-R_B(n_B,t_{1B},t_{2B},m_B,a_B)=&[R_B(n,t_1,t_2,m,a)-R_B(n,t_1,t_{2},m,a_B)]\\
            &+[R_B(n,t_1,t_2,m,a_B)-R_B(n,t_{1B},t_{2B},m,a_B)]\\
            &+[R_B(n,t_{1B},t_{2B},m,a_B)-R_B(n,t_{1B},t_{2B},m_B,a_B)]\\
            &+[R_B(n,t_{1B},t_{2B},m_B,a_B)-R_B(n_B,t_{1B},t_{2B},m_B,a_B)].
        \end{align*}
  From Algorithm A, it follows that   
  $[R_B(n,t_1,t_2,m,a)-R_B(n,t_1,t_{2},m,a_B)]\geq 0$, $[R_B(n,t_1,t_2,m,a_B)-R_B(n,t_{1B},t_{2B},m,a_B)]$ $\geq 0$, $[R_B(n,t_{1B},t_{2B},m,a_B)-R_B(n,t_{1B},t_{2B},m_B,a_B)]\geq 0$, $[R_B(n,t_{1B},t_{2B},m_B,a_B)-R_B(n_B,t_{1B},t_{2B},m_B,a_B)]\geq 0$.
  Therefore, $[R_B(n,t_{1},t_{2},m,a)-R_B(n_B,t_{1B},t_{2B},m_B,a_B)]\geq 0$. Hence, the proof is completed.
    \end{proof}
    
Next, we provide {an upper bound for $n_B$}.
 Let $\boldsymbol{q}_B=(n_B,t_{1B},t_{2B},m_B)$ be the optimal BSPAA. Since all costs are positive and $C_s>v_s\geq 0$, all terms in $R_B(\boldsymbol{q}_B)$ are positive. So we get,
 \begin{align}\label{i1}
     R_B(\boldsymbol{q}_B,a_B)\geq n^*(C_s-v_s).
 \end{align}
 Let $(0,0,0,0)$ denote the no-sampling case, i.e., when we take the decision without life-testing. If the lot is rejected, the Bayes risk is $R(0,0,0,0)=C_r$ and if the lot is accepted, the Bayes risk is $R(0,0,0,0)=E_{\lambda}[h(\lambda)]$. Therefore
 $R(0,0,0,0)=\min\{E_{\lambda}[h(\lambda)], C_r\}$. Now,
 \begin{align}\label{i2}
     R_B(\boldsymbol{q}_B,a_B)\leq \min\{E_{\lambda}[h(\lambda)], C_r\}.
 \end{align}
 From (\ref{i1}) and (\ref{i2}), we get $n_B\leq (\min\{E_{\lambda}[h(\lambda)], C_r\})/(C_s-v_s)$.

\section{Numerical study and comparisons}\label{com}
Here we illustrate the proposed method using an example and compute the optimum BSPAA. Also, we study the effect of cost components and hyperparameters on the optimum solution of BSPAA. The optimum solution is compared with a conventional Bayesian sampling plan through a non-accelerated life test (CBSP) and a conventional Bayesian sampling plan through an accelerated life test (CBSPA). In the numerical example, the loss function is taken as quadratic loss function and the prior distribution is taken as joint gamma and uniform distributions.
\subsection{An Illustrated Example}
The hyperparameters of the prior are taken as $\alpha=3$, $\beta=1$ and $l=10$. The cost coefficients of the loss function are taken as $a_0=2$, $a_1=3$ and $a_2=2$. The values of other cost components are $C_a=0.1$, $v_s=0.2$, $C_s=0.5$, $C_t=5$ and $C_r=30$. The optimal sampling plan $(n_B,t_{1B},t_{2B},m_B)$, the expected time duration $E[\tau_B]$, the expected number of failures $E[D_B]$ and the optimal Bayes risk $R_B$ of the proposed plan are given in Table \ref{t1}.
\begin{table}[hbt]
    \centering
    \caption{Optimal sampling plan for BSPAA}
   \begin{tabular}{|cccc|}
    \hline
$(n_B,t_{1B},t_{2B},m_B)$ &$E[{\tau_B}]$ & $E[D_B]$  &$R_B$\\
\hline
( 3, 0.169, 0.238, 2) &0.220&2.013&27.704\\
\hline
\end{tabular}
\label{t1}
\end{table}

\noindent In our proposed model, if we fix $m=0$, then the model becomes CBSP of Lin et al. \cite{lin2002bayesian} and for $m=n$, the model becomes conventional accelerated BSP (CBSPA). For comparing the CBSP and CBSPA with BSPAA, we calculate the percentage of relative risk savings of a BSPAA over CBSP and CBSPA, which are measured by $RSS_1$ and $RRS_2$, respectively and provided by
\begin{align*}
    RRS_1=100\times \frac{R_1-R_{B}}{R_1}\%
\end{align*}
and 
\begin{align*}
    RRS_2=100\times \frac{R_2-R_{B}}{R_2}\%,
\end{align*}
where $R_1$ and $R_2$ are the optimal Bayes risk of CBSP and CBSPA, respectively. The optimal sampling plan $(n^*,t_1^*)$, the expected time duration $E[\tau^*]$, the expected number of failures $E[D^*]$ and the optimal Bayes risk $R_1$ corresponding to CBSP are provided in Table \ref{t2}. The optimal sampling plan $(n_A,t_{1A},t_{2A})$, the expected time duration $E[\tau_A]$, the expected number of failures $E[D_A]$ and the optimal Bayes risk $R_2$ of CBSPA, and for comparison $RRS_1$ and $RRS_2$ are given in Table \ref{t2}.

\begin{table}[hbt]
    \centering
    \caption{Optimal sampling plan for CBSP and CBSPA, and RRS of BSPAA over CBSP and CBSPA}
   \begin{tabular}{|cccc|cccc|cc|}
    \hline
   \multicolumn{4}{|c|}{CBSP}&\multicolumn{4}{c|}{CBSPA}&\multicolumn{2}{c|}{RRS}\\
    \hline
$(n^*,t_1^*)$ &$E[\tau^*]$& $E[D^*]$& $R_1$&$(n_A,t_{1A},t_{2A})$ &$E[{\tau_A}]$ & $E[D_A]$  &$R_2$&$RRS_1$&$RRS_2$\\
\hline
(4, 0.193) &0.190& 1.644&27.837&(3, 0.162, 0.238)    &0.213&2.163&27.723&0.48\%&0.07\%\\
\hline
\end{tabular}
\label{t2}
\end{table}
In Table \ref{t1}, we see that the Bayes risk of our proposed model is 27.702. However, in Table \ref{t2}, we see that the Bayes risk of CBSP and CBSPA are 27.837 and 27.723, respectively. This indicates that adaptive test may have a better impact than conventional accelerated and non-accelerated life tests. 
\subsection{Effect of the parameters}
Here we study the effect of cost components and hyperparameters on the optimum solution of BSPAA over CBSP and CBSPA. Optimal solutions of BSPAA, CBSP and CBPSA for different values of cost components and hyperparameters for $C_a=0,0.1,0.2$ are provided in Tables \ref{t4}-\ref{t71}. Also, for the comparison of BSPAA, CBSP and CBSPA, the $RSS_1$ and $RSS_2$ are tabulated. The other values of the cost components and hyperparameters, which are not mentioned in the tables, are kept fixed. 

 \begin{sidewaystable}
\small\begin{landscape}
    \centering
       \caption{Optimal sampling parameters of BSPAA, CBSP, and CBSPA and the RSS of BSPAA over CBSP and CBSPA for different values of $C_a$, $\alpha$ and $\beta$}
    \begin{tabular}{|cc|c|cccc|c|cccc|cc|}
    \hline
   && &\multicolumn{4}{c|}{BSPAA}&\multicolumn{1}{c|}{CBSP}&\multicolumn{4}{c|}{CBSPA}&\multicolumn{2}{c|}{RRS}\\
    \hline
$\alpha$&$\beta$ &$C_a$&$(n_B,t_{1B},t_{2B},m_B)$ &$E[{\tau_B}]$ & $E[D_B]$  &$R_B$&{\{$(n^*,t_1^*), E[\tau^*],E[D^*],R_1$\}}&$(n_A,t_{1A},t_{2A})$ &$E[{\tau_A}]$ & $E[D_A]$  &$R_2$&$RRS_1$&$RRS_2$\\
    \hline
         2&0.6& 0&(4, 0.140, 0.185, 2) &0.180&2.214&27.167&\{(4,0.217),& (4, 0.142, 0.188)&0.179&2.480& 27.196&1.26\%&0.11\%\\
          2&0.6& 0.1&(3, 0.202, 0.273, 2) &0.245&2.059&27.359&0.210, 1.843&(3, 0.198, 0.278)   & 0.240&2.223&27.377&0.56\%&0.07\%\\
           2&0.6& 0.2&(3, 0.213, 0.280, 2) &0.251&2.048&27.497&27.514\}& (3, 0.211, 0.286)  &0.248&2.219& 27.544&0.06\%&0.17\%\\
\hline
        2&0.8& 0&(4, 0.090, 0.132, 2) &0.131&1.781& 23.577&\{(5, 0.220)      & (4, 0.093, 0.136)      &0.133&1.933& 23.593&1.28\%&0.07\%\\
 2&0.8&0.1&(3, 0.129, 0.196, 2)  &0.186&1.727&23.811&0.218, 1.781&(3, 0.125, 0.195)   &0.181&1.813&23.824&0.30\%&0.05\%\\
 2&0.8&0.2&(5, 0.220, 0.220, 0)&0.218 &1.781&23.882& 23.882\}&  (3, 0.136, 0.202)   &0.188&1.799&24.046&0.00\%&0.68\%\\
 \hline
 2&1& 0&(2, 0, 0.045, 1-2)   &0.042&0.680& 19.662&\{(0, 0) & (2, 0, 0.045)  &0.042&0.680& 19.662&1.69\%&0.00\%\\
 2&1&0.1&(2, 0, 0.045, 1-2 ) &0.042&0.680&19.862&0, 0, &(2, 0, 0.045)  &0.042&0.680&19.862&0.69\%&0.00\%\\
 2&1&0.2&(0, 0, 0, 0) &0 &0&20&20\}&  (0, 0, 0) &0&0&20&0.00\%&0.00\%\\
 \hline
 2&1.2& 0&(0,0,0,0)&0&0&15.333&\{(0,0), 0, 0, 15.333\}& (0, 0,0)&0&0& 15.333&0.00\%&0.00\%\\
         \hline
         2.5&0.6& 0&(1,0,0.212,1)&0&0&29.844&\{(0,0), 0, 0, 30\}& (0, 0,0)&0&0& 20.40&0.00\%&0.00\%\\
\hline

        2.5&0.8& 0&( 4, 0.130, 0.174, 2)&0.170&2.166& 27.369&\{(5, 0.263) & (4, 0.130, 0.176)    &0.167&2.304& 27.399&1.21\%&0.11\%\\
 2.5&0.8&0.1&(3, 0.186, 0.256, 2) &0.234&2.037&27.572&0.256, 2.543&(3, 0.180, 0.258) &0.227&2.200&27.591&0.47\%&0.07\%\\
 2.5&0.8&0.2&(5, 0.263, 0.263,, 0)&0.256 &2.543&27.703&27.703\}&  (3, 0.194, 0.267) &0.236&2.189&27.768&0.00\%&0.23\%\\
 \hline
 2.5&1& 0&(3, 0.097, 0.169, 3)  &0.157&1.818& 24.156&\{(3, 0.204) & (3, 0.097, 0.169)  &0.157&1.818& 24.156&1.38\%&0.00\%\\
 2.5&1&0.1&(3, 0.113, 0.178, 2 ) &0.169&1.728&24.379&0.198, 1.114&(3, 0.108, 0.176)  &0.164&1.804&24.391&0.47\%&0.05\%\\
 2.5&1&0.2&(5, 0.208, 0.208, 0) &0.198 &1.114&24.493&24.493\}&  (2, 0.099, 0.166) &0.148&1.184&24.579&0.00\%&0.35\%\\
 \hline
 2.5&1.2& 0&(0,0,0,0)&0&0&20.403&\{(0,0), 0, 0, 20.40\}& (0, 0,0)&0&0& 20.40&0.00\%&0.00\%\\
\hline

3&0.8& 0&(0,0,0,0)&0&0&30&\{(0,0), 0, 0, 30\}& (0, 0,0)&0&0& 30&0.00\%&0.00\%\\
\hline
3&1& 0&( 4, 0.119, 0.162, 2) &0.159&2.121& 27.497&\{(4, 0.193), & ( 4, 0.118, 0.163)    &0.156&2.578& 27.526&1.22\%&0.11\%\\
3&1&0.1&( 3, 0.169, 0.238, 2) &0.220&2.013&27.704&0.190, 1.644,&(3, 0.162, 0.238)    &0.213&2.163&27.723&0.48\%&0.07\%\\
3&1&0.2&(4, 0.193, 0.193, 0)&0.190&1.644& 27.837&27.837\}& (2, 0.176 0.266)&0.215&1.516& 27.891&0.00\%&0.19\%\\
\hline
3&1.2& 0&(2, 0, 0.067, 2) &0.057&1.034& 24.486&\{(3, 0.186) & (2, 0, 0.067)   &0.057&1.034& 24.486&1.54\%&0.00\%\\
3&1.2&0.1&(2, 0, 0.067, 2)  &0.057&1.034&24.686&0.182, 1.053&(2, 0, 0.067)&0.057&1.034&24.686&0.74\%&0.00\%\\
3&1.2&0.2&(3, 0.186, 0.186, 0) &0.182&1.053& 24.870& 24.870\}&(2, 0.043, 0.107)    &0.095&1.090& 24.879&0.00\%&0.04\%\\
\hline

3.5&1& 0&(0, 0, 0, 0) &0&0& 30&\{(0,0), 0 0 30\}  &  (0, 0, 0) &0&0& 30&0.00\%&0.00\%\\
\hline
3.5&1.2& 0&(4, 0.108, 0.151, 2) &0.149&2.092&27.569&\{(4, 0.181)  &(3, 0.099, 0.150)    &0.142&1.787&27.588&1.19\%&0.00\%\\
3.5&1.2&0.1&(3, 0.152, 0.220, 2)  
&0.205&1.987&27.775&0.179, 1.554,&(3, 0.145 0.219)   &0.198&2.121&27.792&0.45\%&0.06\%\\
3.5&1.2&0.2&(4, 0.181, 0.181, 0)&0.179&1.554& 27.900& 27.900\}& (2, 0.145, 0.232) &0.191&1.478& 27.929&0.00\%&0.10\%\\
\hline

    \end{tabular}
    \label{t4}
    \end{landscape}
\end{sidewaystable}

\begin{sidewaystable}
\small\begin{landscape}
    \centering
       \caption{Optimal sampling parameters of BSPAA, CBSP, and CBSPA and the RSS of BSPAA over CBSP and CBSPA for different values of $C_a$ and $C_t$ when $v_s=0.2$ }
    \begin{tabular}{|c|c|cccc|c|cccc|cc|}
    \hline
   & &\multicolumn{4}{c|}{BSPAA}&\multicolumn{1}{c|}{CBSP}&\multicolumn{4}{c|}{CBSPA}&\multicolumn{2}{c|}{RRS}\\
    \hline
$C_t$ &$C_a$&$(n_B,t_{1B},t_{2B},m_B)$ &$E[{\tau_B}]$ & $E[D_B]$  &$R_{B}$&{\{$(n^*,t_1^*), E[\tau^*],E[D^*],R_1$\}}&$(n_A,t_{1A},t_{2A})$ &$E[{\tau_A}]$ & $E[D_A]$  &$R_2$&$RRS_1$&$RRS_2$\\
         \hline
        1& 0&(3, 0.525, 0.579, 3) &0.449&2.486& 26.409&\{(3, 0.752) ,& (3, 0.525, 0.579)&0.449&2.486& 26.409&0.18\%&0.00\%\\
 1&0.1&(3, 0.752, 0.752, 0)  &0.532&2.442&26.457&0.532, 2.442, 26.457\}&(3, 0.601, 0.639) &0.481&2.485& 26.488&0.00\%&0.12\%\\
\hline

2& 0&(4, 0.266, 0.321, 3)  &0.300&2.808&26.805&\{(4, 0.394)& (3, 0.339, 0.424) &0.351&2.442& 26.811&0.53\%&0.02\%\\
2& 0.1&(3, 0.413, 0.484, 3)  &0.393&2.455&26.926&0.360, 2.523, & (3, 0.413, 0.484)&0.393&2.455& 26.926&0.08\%&0.00\%\\
2& 0.2&(4, 0.394, 0.394, 0)  &0.360&2.523&26.948&26.948\}& (3, 0.483, 0.539) &0.428&2.456& 27.025&0.00\%&0.28\%\\
\hline
5& 0&( 4, 0.119, 0.162, 2) &0.159&2.121& 27.497&\{(4, 0.193), & ( 4, 0.118, 0.163)    &0.156&2.578& 27.526&1.22\%&0.11\%\\
5&0.1&( 3, 0.169, 0.238, 2) &0.220&2.013&27.704&0.190, 1.644,&(3, 0.162, 0.238)    &0.213&2.163&27.723&0.48\%&0.07\%\\
5&0.2&(4, 0.193, 0.193, 0)&0.190&1.644& 27.837&27.837\}& (2, 0.176 0.266)&0.215&1.516& 27.891&0.00\%&0.19\%\\
\hline
7& 0&(4, 0.106, 0.152, 2) &0.149&2.138& 27.806&\{(4, 0.185),  & (4, 0.103, 0.150) &0.145&2.336& 27.827&1.43\%&0.08\%\\
7&0.1&(3, 0.103, 0.153,  2)  &0.146&1.726&28.049&0.182, 1.596 &(3, 0.097, 0.151)&0.141&1.822&28.061&0.57\%&0.04\%\\
7&0.2&(4, 0.185, 0.185, 0)&0.182&1.596& 28.209& 28.209\}& (4, 0.155, 0.186)&0.181&2.243& 28.493&0.00\%&1.00\%\\
\hline
   \end{tabular}
    \label{t5}
\end{landscape}
\small\begin{landscape}
    \centering
       \caption{Optimal sampling parameters of BSPAA, CBSP, and CBSPA and the RSS of BSPAA over CBSP and CBSPA for different values of $C_a$ and $C_t$ when $v_s=0$}
    \begin{tabular}{|c|c|cccc|c|cccc|cc|}
    \hline
     & &\multicolumn{4}{c|}{BSPAA}&\multicolumn{1}{c|}{CBSP}&\multicolumn{4}{c|}{CBSPA}&\multicolumn{2}{c|}{RRS}\\
    \hline
$C_t$ &$C_a$&$(n_B,t_{1B},t_{2B},m_B)$ &$E[{\tau_B}]$ & $E[D_B]$  &$R_{B}$&{\{$(n^*,t_1^*), E[\tau^*],E[D^*],R_1$\}}&$(n_A,t_{1A},t_{2A})$ &$E[{\tau_A}]$ & $E[D_A]$  &$R_2$&$RRS_1$&$RRS_2$\\
\hline
         1& 0&(3, 0.518, 0.592, 3) &0.453&2.553&26.506&(3, 0.790) &(3, 0.518, 0.592)&0.453&2.553&26.506&0.22\%&0.00\%\\
1&0.1&(3, 0.602, 0.656, 3)&0.487&2.547&26.585&0.546, 2.477&(3, 0.602, 0.656)   &0.487&2.547&26.585&0.08\%&0.00\%\\
1&0.2&(3, 0.790, 0.790, 0)&0.546&2.477&26.565&26.565\}&(3, 0.691, 0.727)   &0.520&2.548&26.653&0.00\%&0.33\%\\
\hline
2& 0&(3, 0.340, 0.453, 3) &0.362&2.543&26.912&(3, 0.713),  &(3, 0.340, 0.453)&0.362&2.543&26.912&0.68\%&0.00\%\\
2&0.1&(3, 0.412, 0.503, 3) &0.399&2.529&27.028&0.516, 2.403, &(3, 0.412, 0.503)&0.399&2.529&27.028&0.25\%&0.00\%\\
2&0.2&(3, 0.713, 0.713,0)&0.516&2.403&27.096&27.096\}&(3, 0.481, 0.553)  &0.440&2.199&27.127&0.00\%&0.11\%\\
\hline
5& 0&(3, 0.142, 0.239, 3) &0.206&2.266&27.683&\{(3, 0.279)  &(3, 0.142, 0.239)&0.206&2.266&27.683&1.71\%&0.00\%\\
5&0.1&(2, 0.149, 0.255, 2) &0.200&1.549&27.857&0.259, 1.566, &(2, 0.149, 0.255)&0.200&1.549&27.857&1.09\%&0.00\%\\
5&0.2&(2, 0.172, 0.270, 2) &0.215&1.542&27.986&28.165\}&(2, 0.172, 0.270)&0.215&1.542&27.986&0.63\%&0.00\%\\
\hline
7& 0&(2, 0.027, 0.129, 2) &0.098&1.405&28.009&\{(4, 0.190) &(2, 0.027, 0.129)&0.098&1.405&28.009&2.36\%&0.00\%\\
7&0.1&(2, 0.034, 0.132, 2) &0.102&1.395&28.192&0.187, 1.626,&(2, 0.034, 0.132)&0.102&1.395&28.192&1.172\%&0.00\%\\
7&0.2&(1, 0, 0.083, 1) &0.050&0.629&28.297&28.686\}&(1, 0, 0.083)&0.050&0.629&28.297&1.36\%&0.00\%\\
\hline

      \end{tabular}
    \label{t51}
\end{landscape}
\end{sidewaystable}
\begin{sidewaystable}
\small\begin{landscape}
    \centering
       \caption{Optimal sampling parameters of BSPAA, CBSP, and CBSPA and the RSS of BSPAA over CBSP and CBSPA for different values of $C_a$ and $C_r$}
    \begin{tabular}{|c|c|cccc|c|cccc|cc|}
    \hline
   & &\multicolumn{4}{c|}{BSPAA}&\multicolumn{1}{c|}{CBSP}&\multicolumn{4}{c|}{CBSPA}&\multicolumn{2}{c|}{RRS}\\
    \hline
$C_r$ &$C_a$&$(n_B,t_{1B},t_{2B},m_B)$ &$E[{\tau_B}]$ & $E[D_B]$  &$R_{B}$&{\{$(n^*,t_1^*), E[\tau^*],E[D^*],R_1$\}}&$(n_A,t_{1A},t_{2A})$ &$E[{\tau_A}]$ & $E[D_A]$  &$R_2$&$RRS_1$&$RRS_2$\\
         \hline
        10& 0&(0,0,0,0)&0&0& 10&\{(0, 0) ,0 ,0, 10\}& (0, 0, 0)&0&0& 10&0.00\%&0.00\%\\
\hline
20& 0&(0,0,0,0)&0&0& 10&\{(0, 0) ,0 ,0, 10\}& (0, 0, 0)&0&0& 10&0.00\%&0.00\%\\
\hline
30& 0&(3, 0.200, 0.263, 3)&0.237&2.132& 27.621&\{(3, 0.272), & (3, 0.200,0.263)&0.237&2.132& 27.621&0.91\%&0.00\%\\
30&0.1&(3,0.204,0.260,2)&0.241&1.965&27.700&0.252, 1.542,&(3,0.207,0.269)&0.242&2.136&27.793&0.63\%&0.33\%\\
30&0.2&(3, 0.227, 0.272, 0)&0.252&1.542& 27.875&27.875\}& (3, 0.241,0.327)&0.281&2.338& 27.950&0.00\%&0.27\%\\
\hline
40& 0&(5, 0.119, 0.154, 3)  &0.153&2.604&31.537&\{(5, 0.179)  & (4, 0.119, 0.163)&0.158&2.344& 31.550&0.92\%&0.04\%\\
40& 0.1&(4, 0.148, 0.186, 3)  &0.181&2.298&31.812& 0.178, 1.949, & (4, 0.141, 0.180) &0.175&2.339& 31.828&0.06\%&0.05\%\\
40& 0.2&(5, 0.179, 0.179, 0)  &0.178&1.949&31.830& 31.830\}& (4, 0.175, 0.219) &0.210&2.538& 32.084&0.00\%&0.79\%\\
\hline
50& 0&(5, 0.114, 0.148, 4) &0.146&2.649& 34.043&\{(4, 0.169)   & (5, 0.114, 0.147)    &0.145&2.639& 34.045&0.83\%&0.01\%\\
50&0.1&(4, 0.115, 0.148, 3)&0.146&2.084&34.342&0.167, 1.496, &(4,0.119, 0.157)&0.153&2.234&34.349&0.03\%&0.02\%\\
50&0.2&(4, 0.169, 0.169, 0)&0.167&1.496& 34.329& 34.329\}& (4, 0.113, 0.152)&0.148&2.229&34.344&0.00\%&0.04\%\\
\hline
60& 0&(0,0,0,0) &0&0& 35&\{(0, 0), 0 0 35\}  &  (0, 0, 0) &0&0&  35&0.00\%&0.00\%\\
\hline
  \end{tabular}
    \label{t6}  
\end{landscape}
\small\begin{landscape}
    \centering
       \caption{Optimal sampling parameters of BSPAA, CBSP, and CBSPA and the RSS of BSPAA over CBSP and CBSPA for different values of $C_a$ and $l$}
    \begin{tabular}{|c|c|cccc|c|cccc|cc|}
    \hline
   & &\multicolumn{4}{c|}{BSPAA}&\multicolumn{1}{c|}{CBSP}&\multicolumn{4}{c|}{CBSPA}&\multicolumn{2}{c|}{RRS}\\
    \hline
$l$ &$C_a$&$(n_B,t_{1B},t_{2B},m_B)$ &$E[{\tau_B}]$ & $E[D_B]$  &$R_{B}$&{\{$(n^*,t_1^*), E[\tau^*],E[D^*],R_1$\}}&$(n_A,t_{1A},t_{2A})$ &$E[{\tau_A}]$ & $E[D_A]$  &$R_2$&$RRS_1$&$RRS_2$\\
         \hline
10& 0&(3, 0.200, 0.263, 3)&0.237&2.132& 27.621&\{(3, 0.272), & (3, 0.200,0.263)&0.237&2.132& 27.621&0.91\%&0.00\%\\
10&0.1&(3,0.204,0.260,2)&0.241&1.965&27.700&0.252, 1.542,&(3,0.207,0.269)&0.242&2.136&27.793&0.63\%&0.33\%\\
10&0.2&(3, 0.227, 0.272, 0)&0.252&1.542& 27.875&27.875\}& (3, 0.241,0.327)&0.281&2.338& 27.950&0.00\%&0.27\%\\
\hline
20& 0&(3, 0.161, 0.203, 2)   &0.193&1.755&27.462&\{(3,0.272)    & ( 3, 0.158, 0.202)  &0.190&1.862& 27.469&1.48\%&0.02\%\\
20& 0.1&(3, 0.174, 0.213, 2)   &0.200&2.017&27.626& 0.252, 1.542, & (3, 0.168, 0.210) &0.193&2.182& 27.659&0.89\%&0.12\%\\
20& 0.2&(3, 0.185, 0.221, 2)  &0.208&1.994&27.783& 27.875\}& (3, 0.179, 0.219)  &0.202&2.176& 27.845&0.33\%&0.22\%\\
\hline
50& 0&(3, 0.164, 0.182, 2) &0.176&1.654&  27.404&\{(3, 0.0.272)   & (3, 0.162, 0.181)    &0.174&1.768& 27.424&0.33&0.22\%\\
50&0.1&(3, 0.174, 0.192, 2) &0.185&1.678&27.566&0.252, 1.542,&(3, 0.171, 0.189) &0.181&1.769&27.613&1.11\%&0.17\%\\
50&0.2&(3, 0.185, 0.221, 2) &0.208&1.994& 27.783& 27.875\}& ( 3, 0.179, 0.219) &0.202&2.176&27.845&0.33\%&0.22\%\\
\hline
100& 0&( 3, 0.162, 0.172, 3) &0.167&1.508& 27.407&\{( 3, 0.272 )  &  ( 3, 0.162, 0.172) &0.167&1.508& 27.407&1.68\%&0.00\%\\
100& 0.1&(2, 0.117, 0.342, 2)  &0.179&1.877& 27.536&0.252, 1.542  & (2, 0.117, 0.342)  &0.179&1.877& 27.536&1.22\%&0.00\%\\
100& 0&(2, 0.131, 0.357, 2)  &0.190&1.880& 27.678&27.875\} &  (2, 0.131, 0.357)  &0.190&1.880& 27.678&0.70\%&0.00\%\\
\hline

    \end{tabular}
    \label{t61}  
\end{landscape}
\end{sidewaystable}
\begin{sidewaystable}
\begin{landscape}
 \centering
       \caption{Optimal sampling parameters of BSPAA, CBSP, and CBSPA and the RSS of BSPAA over CBSP and CBSPA for different values of $C_a$ and $a_0$ }
    \begin{tabular}{|c|c|cccc|c|cccc|cc|}
    \hline
   & &\multicolumn{4}{c|}{BSPAA}&\multicolumn{1}{c|}{CBSP}&\multicolumn{4}{c|}{CBSPA}&\multicolumn{2}{c|}{RRS}\\
    \hline
$a_0$ &$C_a$&$(n_B,t_{1B},t_{2B},m_B)$ &$E[{\tau}]$ & $E[D]$  &$R_B$&\{$(n^*,t_1^*), E[\tau],E[D],R_1$\}&$(n_A,t_{1A},t_{2A})$ &$E[{\tau}]$ & $E[D]$  &$R_2$&$RRS_1$&$RRS_2$\\
         \hline
         0& 0&(3, 0.154, 0.221, 2)   &0.205&1.978&26.685&\{(5, 0.235), &(3, 0.150, 0.222)     &0.200&2.112&26.704&0.33\%&0.07\%\\
0&0.1&(4, 0.108, 0.149, 2)
 &0.147&2.060& 26.480&0.231, 2.346,&(3, 0.138, 0.214)&0.192&2.121&26.503&1.10\%&0.09\%\\
0&0.2&(5, 0.235, 0.235, 0)  &0.231&2.346&26.775&26.775\}&(3, 0.163, 0.230)   &0.208&2.096&26.898&0.00\%&0.46\%\\
\hline
         1& 0&(4, 0.113, 0.155, 2)   &0.152&2.090&26.998&\{(5, 0.243), &(4, 0.115, 0.157)     &&&27.025&1.18\%&0.10\%\\
1&0.1&(3, 0.162, 0.230, 2)
 &0.213&1.996& 27.203&0.239, 2.397,&(3, 0.156, 0.230)&0.206&2.138&27.223&0.43\%&0.07\%\\
1&0.2&(5, 0.243, 0.243, 0)  &0.239&2.397&27.321&27.321\}&(3, 0.170, 0.239)   &0.215&2.125&27.414&0.00\%&0.34\%\\
\hline
2& 0&(3, 0.200, 0.263, 3)&0.237&2.132& 27.621&\{(3, 0.272), & (3, 0.200,0.263)&0.237&2.132& 27.621&0.91\%&0.00\%\\
2&0.1&(3,0.204,0.260,2)&0.241&1.965&27.700&0.252, 1.542,&(3,0.207,0.269)&0.242&2.136&27.793&0.63\%&0.33\%\\
2&0.2&(3, 0.227, 0.272, 0)&0.252&1.542& 27.875&27.875\}& (3, 0.241,0.327)&0.281&2.338& 27.950&0.00\%&0.27\%\\
\hline
3&0&(4, 0.125, 0.170, 2) &0.167&2.165&27.978&\{(4, 0.201) &(3, 0.116, 0.172) &0.160&1.892&28.004&1.15\%&0.09\%\\
3&0.1&(3, 0.175, 0.246, 2) &0.226&2.035&28.185&0.197, 1.691, &(3, 0.167, 0.246) &0.218&2.191&28.203&0.43\%&0.06\%\\
3& 0.2&(4, 0.201, 0.201, 0) &0.197&1.691&28.306&28.306\}&(3, 0.182, 0.256) &0.228&2.182&28.388&0.00\%&0.29\%\\

\hline
5& 0&(3, 0.246, 0.338, 3) &0.288&2.374&29.547&\{(3, 0.354), &(3,0.246, 0.338)&0.288&2.374&29.547&1.19\%&0.00\%\\
5&0.1&(2, 0.341, 0.439, 2) &0.323&1.663&29.656&0.316, 1.791,&(2, 0.341, 0.439)&0.323&1.663&29.656&0.82\%&0.00\%\\
5&0.2&(2, 0.350, 0.445, 2)&0.327&1.661&29.738&29.902\}&(2, 0.350, 0.445)&0.327&1.661&29.738&0.55\%&0.00\%\\
\hline
10& 0&(0,0,0,0)&0&0&30&(0,0),0,0,30&(0,0,0)&0&0&30&0.00\%&0.00\%\\
\hline
    \end{tabular}
    \label{t7}
    \end{landscape}
\begin{landscape}
    \centering
       \caption{Optimal sampling parameters of BSPAA, CBSP, and CBSPA and the RSS of BSPAA over CBSP and CBSPA for different values of $C_a$ and $a_1$}
    \begin{tabular}{|c|c|cccc|c|cccc|cc|}
    \hline
   & &\multicolumn{4}{c|}{BSPAA}&\multicolumn{1}{c|}{CBSP}&\multicolumn{4}{c|}{CBSPA}&\multicolumn{2}{c|}{RRS}\\
    \hline
$a_1$ &$C_a$&$(n_B,t_{1B},t_{2B},m_B)$ &$E[{\tau}]$ & $E[D]$  &$R_B$&\{$(n^*,t_1^*), E[\tau],E[D],R_1$\}&$(n_A,t_{1A},t_{2A})$ &$E[{\tau}]$ & $E[D]$  &$R_2$&$RRS_1$&$RRS_2$\\
         \hline
0& 0&(4, 0.109, 0.150, 3) &0.147&2.219& 23.462&\{(5, 0.172)  & (4, 0.108, 0.149) &0.145&2.245& 23.464&0.02\%&1.21\%\\
0&0.1&(3, 0.096, 0.147, 2) &0.140&1.724&23.704&0.223, 1.407,&(3, 0.094, 0.148) &0.139&1.814&23.718&0.19\%&0.06\%\\
0&0.2&(3, 0.172, 0.172 0)&0.223&1.407& 23.750&23.750\}& (3, 0.105, 0.155)   &0.146&1.795& 23.944&0.00\%&0.81\%\\
\hline
1& 0&(4, 0.127, 0.174, 4)  &0.168&2.423& 25.025&\{(5, 0.198)   &(4, 0.127, 0.174)  &0.168&2.423& 25.025&1.12\%&0.00\%\\
1&0.1&(3, 0.120, 0.178, 2)  &0.168&1.846&25.250&0.196, 2.092,&(3, 0.118, 0.178)  &0.165&1.939&25.267&0.24\%&0.07\%\\
1&0.2&(5, 0.198, 0.198, 0)&0.196&2.092& 25.311&25.311\}& (3, 0.129, 0.186)    &0.173&1.934& 25.478&0.00\%&0.65\%\\
\hline
2& 0.1&(4, 0.101, 0.141, 2)  &0.139&2.025& 26.375&\{(3, 0.272), & (3, 0.141, 0.209) &0.190&2.062& 26.593&1.15\%&0.08\%\\
2& 0.1&(3, 0.145, 0.209, 2) &0.195&1.940& 26.575&0.252, 1.542, & (3, 0.141, 0.209) &0.190&2.062& 26.593&1.21\%&0.09\%\\
2&0.2&(3, 0.272, 0.272, 0)&0.252&1.542& 26.775&26.775\}& (3, 0.163, 0.230)&0.208&2.122& 26.898&0.00\%&0.36\%\\
\hline
3& 0&(3, 0.200, 0.263, 3)&0.237&2.132& 27.621&\{(3, 0.272), & (3, 0.200,0.263)&0.237&2.132& 27.621&0.91\%&0.00\%\\
3&0.1&(3,0.204,0.260,2)&0.241&1.965&27.700&0.252, 1.542,&(3,0.207,0.269)&0.242&2.136&27.793&0.63\%&0.33\%\\
3&0.2&(3, 0.227, 0.272, 0)&0.252&1.542& 27.875&27.875\}& (3, 0.241,0.327)&0.281&2.338& 27.950&0.00\%&0.27\%\\
\hline
5& 0&(0,0,0,0)&0&0&30&(0,0),0,0,30&(0,0,0)&0&0&30&0.00\%&0.00\%\\
\hline
\end{tabular}
\label{t71}
\end{landscape}
\end{sidewaystable}

In Table \ref{t4}, the hyperparameters $\alpha$ and $\beta$ vary while other values of the cost components and hyperparameters are fixed. For $(\alpha,\beta)=(2,1.2),(2.5,0.6),(2.5,1.2),(3,0.8),(3.5,1)$, we observe that the optimal BSPs represent no sampling cases. It is observed that for fixed $\alpha=2$ and $C_a$, when $\beta$ increases from 0.6 to 1, the Bayes risk and expected time duration $E[\tau_B]$ of BSPAA 
decrease.  This is due to the fact that when $\beta$ increases, the prior mean $\beta/(\alpha-1)$ increases. Also, it is seen that for fixed $\beta=1$ and $C_a=0,0.1,0.2$, when $\alpha$ increases from 2 to 3, the Bayes risk and expected time duration $E[\tau_B]$ of BSPAA  increase. It is seen that when $C_a$ increases, $RRS_1$ decreases. Therefore, for the higher values of $c_a$, CBSP is better than CBSPA and BSPAA is equivalent to CBSP.

In Table \ref{t5}, we provide the effect of $C_t$ when $c_a=0,0.1,0.2$, $v_s=0$ for fixed values of other parameters and coefficients. In Table \ref{t51}, we provide the effect of $C_t$ when $v_s=0.2$. It is seen that when $c_t$ increases, the expected time duration decreases as expected. When $C_a=0$ and $v_s=0$, it is seen that in the optimal sampling plan of BSPAA, $m_B=n_B$. This means that the optimal sampling plan of BSPAA is equivalent to the optimal sampling plan of CBSPA. Therefore $RSS_2=0\%$. This is due to the fact that when $C_a$ and $v_s$ are not incorporated into the Bayes risk, the number of failures in the life test does not depend on the sampling plan. In CBSPA, we get more information in a shorter time duration.  Therefore, BSPAA is equivalent to CBSPA.

In Table \ref{t6}, we provide the effect of $C_r$. For $C_r=10,20,60$, the optimal BSPs represent no sampling cases. For $C_r=10,20$, the Bayes risk is $R_B=C_r$, that is the lot is rejected without life testing. For $C_r=60$, $R_B=a_0+a_1\alpha/\beta+a_2\alpha(\alpha+1)/\beta^2$ and the lot is accepted without life testing. For $C_r=30,40,50$, the optimal sampling parameters have sampling cases. 

In Table \ref{t61}, we provide the effect of $l$. When $l$ increases, the Bayes risk decreases for fixed $C_a$. This is due to the fact that when $l$ increases, the mean accelerated factor $\phi$ increases. We get more information in a shorter period of time. Due to similar facts, the expected time duration decreases when $l$ increases for fixed $C_a$. Since the CBSP does not depend on $l$, the optimal sampling parameters remain fixed for different values of $l$.

 The effects of $a_0$ and $a_1$ are provided Table \ref{t7} and \ref{t71}, respectively. we provide the effect of $a_1$. It is seen that when $a_i$, $i=0,1$ increases. the Bayes risk increases. For high values of $a_i$, $i=0,1$, the optimal BSPs are no sampling cases, that is the lot is rejected without life testing.
 
From Tables \ref{t4}-\ref{t71}, it is observed that  for sampling cases, if $n_B=n^*=n_A$, $E[\tau^*]\leq E[\tau_B]\leq E[\tau_A]$ and also $E[D^*]\geq E[D_B]\geq E[D_A]$. This demonstrates that the ordering of the duration is the opposite of the ordering of the expected number of failures. Therefore, when $C_t$, $v_s$ and $C_a$ are incorporated into Bayes risk, the BSPAA is better than the other two sampling plans, which can be observed from the values of $RRS_1$ and $RRS_2$. It is observed that when $0<m_B<n$, we find the $RSS_1$ is greater than $1.2\%$ but for $RRS_2$, we find some values greater than $0.15\%$. As CBSP does not depend on $C_a$, only one value is provided in Tables \ref{t4}-\ref{t71} when $C_a$ varies.

\section{Data analysis}\label{nu}
The proposed methodology of determining optimum BSP is illustrated by the data set of oil
breakdown times of insulating fluid subjected to different constant levels of high voltage reported in \cite{lawless2011statistical}. For the sake of illustration,  30kV is assumed to be the normal use voltage and 36kV is used as the accelerated voltage. The data is given in Table \ref{t8}. Zheng \& Fang \cite{zheng2017exact} analyzed the data and showed that the exponential distribution fits well.

The mean oil breakdown time of insulating fluid at normal conditions and accelerated conditions are 75.78 and 4.61 respectively. The MLEs of $\lambda$ and $\phi$ are $\hat{\lambda}=0.013$ and $\hat{\phi}=16.45$ respectively. For computational purposes, we assume that $\lambda$ has gamma prior with mean 0.013. The hyperparameters of gamma prior of $\lambda$ are taken as $\alpha=1.3$ and $\beta=100$, and the hyperparameter of uniform prior of $\phi$ are taken as $l=30$.  
\begin{table}[hbt!]
    \centering
     \caption{The data set of oil
breakdown times}
    \begin{tabular}{l}
    \hline
\textbf{Normal conditions (30kV):}\\
7.74, 17.05, 20.46, 21.02, 22.66, 43.40, 47.30, 139.07, 144.12, 175.88, 194.90\\
\hline
\textbf{Accelerated stress condition (36kV): }\\
1.97, 0.59, 2.58, 1.69, 2.71, 25.50, 0.35, 0.99, 3.99, 3.67, 2.07, 0.96, 5.35, 2.90, 13.77\\
\hline
    \end{tabular}
    \label{t8}
\end{table}

We consider the cost components of the manufacturer, which are taken as $C_a=0.1$, $v_s=0.2$, $c_s=0.4$, $C_t=0.05$, $a_0=2$, $a_1=700$, $a_2=80000$ and $C_r=30$. The optimum RASP is obtained as $(n_b,t_{1B},t_{2B},m_B)=(4, 18.29, 28.29, 2)$. Next, we have to carry out a life test under the optimum RASP. For illustration, we generate Type-I adaptive step-stress data based on the optimum life testing plan $(n_B,t_{1B},t_{2B},m_B)=(4, 18.29, 28.29, 2)$ from the exponential distribution with $\lambda=0.013$ and $\phi=16.45$. The data sets $(\boldsymbol{y},d_1,d_2)=(y_1,y_2,y_3,y_4,d_1,d_2)$ and the corresponding decisions about the lot acceptance or rejection $a_B$ are given in Table \ref{ttt} for illustration purposes. Also, in Table \ref{ttt}, the decision of changing the stress at $t_{1B}$, $w_1$  and $w_2$ are provided. 

\begin{table}[hbt!] 
    \centering
     \caption{Simulated data sets and corresponding decision about the lot }
   \resizebox{\textwidth}{!}{ \begin{tabular}{|c|ccccccccccc|}
    \hline
       $i$& $y_1$ &$y_2$ &$y_3$ &$y_4$& $d_1$& Change the stress & $d_2$& $w_1$ &$w_2$& $\varphi(\boldsymbol{y},d_1,d_2)-C_r$& $a_B$ \\
        \hline
     1&  18.76& 19.58& 20.00& 23.56& 0& Yes & 4& 73.16& 8.75& 4.33 &0\\
     2&   6.83&   7.97&  24.72&-& 2& No& 1 &51.39& 16.43&54.67& 0\\
       3&  10.20& 19.44& 20.02& -&1&Yes&2& 65.07& 12.88& -2.13&1\\
4&5.97 & 7.90& -&-&2& No &0& 50.45& 20.00& 24.62&0\\
  5&      15.62& 18.98& 19.74&21.78&1&Yes&3&70.49& 5.64& 19.37&0\\
 6&  20.09& 20.58& 21.81& 23.49&0&Yes&4& 73.16& 12.81&-3.31&1\\
  7&  14.84& 21.55&21.75& 28.11&1&Yes&3&69.71& 16.54&-1.57&1\\
        \hline
    \end{tabular}
   }
    \label{ttt}
\end{table}
The BSP can be illustrated as follows. Four items are put on life test at the stress level 30kV. After time $t_1=18.29$, the number of failures is observed. If the number of failures is less than $2$, the stress is increased to 36kV and the test continues up to $t_2=28.29$. If the number of failures is greater than equals to $2$, the stress is unchanged and the test continues up to $t_2=28.29$ at the stress level 30kV.   Then we calculate  $e_i=\varphi(y_i,d)-C_r$, for $i=1,\ldots,7$ If  $e_i<0$, then $a_B=1$ and the lot is accepted. If $e_i\geq 0$ then $a_B=0$ and the lot is rejected.

 \section{Conclusion}\label{con}
This work considered designing an adaptive BSP based on a simple step-stress test for type-I censored data. It is seen that when an additional cost for increasing stress levels and salvage costs are incorporated, the BSPAA is more effective than the other two sampling plans for type-I censored data. When an additional cost for increasing stress levels is relatively very high, the BSP is better than CBSPA. When $C_a=0$ and $v_s=0$, the CBSPA is always better than any other sampling plan. In that case, BSPAA is equivalent to CBSPA. In this work we have considered exponential distribution for illustration. However, the proposed method can be extended for other lifetime distributions. The work can also be extended to other censoring schemes and more than 2-stage step-stress test.

Tsai et al. \cite{tsai2014efficient} studied an efficient sampling plan for type-I censored data where the manufacturer can take a decision and terminate the test before completing the life test. The procedure is known as the curtailment procedure. Chen et al. \cite{chen2017curtailed} studied the curtailment BSP for type-II censored data. Then Chen et al. \cite{chen2023designing} studied the curtailment BSP based on simple step stress ALT for type-II censored data. Similarly, this work can be extended to the curtailment of BSP.  
\bibliographystyle{apalike}
\bibliography{citation}

\appendix
\section{D{etailed expression of} \texorpdfstring{$E[D]$, ${n_{as}}$ {and} $E[\tau]$}{}:}\label{appB}
\subsection{Expected number of failures (\texorpdfstring{$E[D]$}{}): } The total number of failures during the life test is defined as $D=D_1+D_2$. Expected number of failures is given by 
\begin{align*}
    E[D]=\int_{\boldsymbol{\theta}}  E[D_1+D_2 \ | \ \boldsymbol{\theta}] p(\boldsymbol{\theta})~d\boldsymbol{\theta},
\end{align*}
where
\begin{align*}
    E[D_1+D_2\ | \ \boldsymbol{\theta}]&=E[D_1 \ | \ \boldsymbol{\theta}]+E[D_2\ | \ \boldsymbol{\theta}]\\
    &=E[D_1\ | \boldsymbol{\theta}]+\sum_{d_1=0}^nE[D_2\ | \ D_1=d_1]P(D_1=d_1)\\
&=n[1-\exp(-\lambda t_1)]+\sum_{d_1=0}^{n} (n-d_1)\left(1-\frac{R(t_2)}{R(t_1)}\right)\binom{n}{d_1}(1-R(t_1))^{d_1}(R(t_1))^{n-d_1}\\
&=n(1-\exp(-\lambda t_1)+n\exp(-\lambda t_1)-\left[\sum_{d_1=0}^{m-1} (n-d_1)\exp[-\lambda (t_1+\phi(t_2-t_1)]+\sum_{d_1=m}^{n}(n- d_1)\exp(-\lambda t_2)\right]\\
&~~~~~~~~~~~~~~~~~~~~~~~~~~~~~~~~~~~~~~~~~~~~~~~~~~~~~~~~~~~~~\binom{n}{d_1}(1-\exp(-\lambda t_1))^{d_1}\exp[-\lambda(n-d_1-1)t_1]\\
&=n-\sum_{d_1=0}^{m-1}\sum_{i=0}^{d_1}(-1)^{d_1-i}\binom{d_1}{i} (n-d_1)\exp[-\lambda ((n-i)t_1+\phi(t_2-t_1)]\\
&~~~~~~~~-\sum_{d_1=m}^{n}\sum_{i=0}^{d_1}(-1)^{d_1-i}\binom{d_1}{i}(n- d_1)\exp[-\lambda((n-i)t_1+(t_2-t_1)) ].
\end{align*}

\subsection{Expected number of items continued to higher stress levels after \texorpdfstring{$t_1$} (\texorpdfstring{$n_{as}$}): }
Let $n_{as}$ be the expected number of items continued to higher stress levels after $t_1$. Then
\begin{align*}
    n_{as}=E_{\boldsymbol{\theta}}[(n-D_1)\ | \ D_1<m,\boldsymbol{\theta}]&=\int_{\lambda=0}^\infty\sum_{d_1=0}^{m-1}(n-d_1)\binom{n}{d_1}(1-\exp(-\lambda t_1))^{d_1}(\exp(-(n-d_1)t_1))~p_1(\lambda)d\lambda\\ 
    &=\sum_{d_1=0}^{m-1}\sum_{i=0}^{d_1}(n-d_1)\binom{n}{d_1}\binom{d_1}{i}(-1)^{d_1-i}\int_{\lambda=0}^\infty(\exp(-(n-i)t_1))~p_1(\lambda)~d\lambda
\end{align*}

\subsection{ Expected test duration (\texorpdfstring{$E[\tau]$}{}): }
The duration of the test under type-I censoring is defined as
$\tau=\min\{Z_n,t_2\}$.
\allowdisplaybreaks{\begin{align*}
E[\tau\ | \ \boldsymbol{\theta}]=&E[\min\{Z_n,t_2\}]\\
=&E[t_2\ | Z_n\geq t_2]P(Z_n>t_2)+E[Z_n\ |\ Z_n<t_2][1-P(Z_n>t_2)]\\
=&t_2R_{Z_n}(t_2)]-\int_{t=0}^{t_2}t\frac{\partial}{\partial t}R_{Z_n}(t)]~ dt\\
=&t_2R_{Z_n}(t_2)-t_2R_{Z_n}(t_2)+\int_{t=0}^{t_2}R_{Z_n}(t)~ dt\\
=&\int_{t=0}^{t_2}R_{Z_n}(t)~ dt
\end{align*}}
When $m> 0$, for $0<t<t_1$,
\begin{align*}
    R_{Z_n}(t)=1-P(Z_n<t)=1-P(\max\{Y_1,\ldots,Y_n\}\leq t\}=1-[1-R(t)]^n&=\sum_{i=1}^{n}\binom{n}{i}(-1)^{i+1}[R(t)]^i\\
    &=\sum_{i=1}^{n}\binom{n}{i}(-1)^{i+1}\exp(-\lambda i t)
\end{align*}
 and for $t_1<t<t_2$, the distribution $Z_n$ depends on $D_1$.
\begin{align*}
R_{Z_n}(t)=&\sum_{d_1=0}^{n-1}P(Z_n>t\ |\ D_1=d_1)P(D_1=d_1)\\
  =&\sum_{d_1=0}^{n-1}[1-P(Z_n<t\ |\ D_1=d_1)][P(D_1=d_1)]\\
  =&\sum_{d_1=0}^{n-1}[1-P(\text{$(n-d_1)$ items fails in the interval $(t_1,t)$})]\binom{n}{d_1}[1-R(t_1)]^{d_1}[R(t_1)]^{n-d_1}\\
  =&\sum_{d_1=0}^{n-1}\left[1-\left(1-\frac{R(t)}{R(t_1)}\right)^{n-d_1}\right]\binom{n}{d_1}\sum_{j=0}^{d_1}\binom{d_1}{j}(-1)^j(R(t_1))^{n-d_1+j}\\
  =&-\sum_{d_1=0}^{n-1}\sum_{j=0}^{d_1}\sum_{k=1}^{n-d_1}\binom{n}{d_1}\binom{d_1}{j}\binom{n-d_1}{k}(-1)^{j+k}\left(\frac{R(t)}{R(t_1)}\right)^{k}(R(t_1))^{n-d_1+j}\\
  =&-\sum_{d_1=0}^{m-1}\sum_{j=0}^{d_1}\sum_{k=1}^{n-d_1}\binom{n}{d_1}\binom{d_1}{j}\binom{n-d_1}{k}(-1)^{j+k}\exp[-\lambda((n-d_1+j)t_1+\phi k(t-t_1))]\\
  &-\sum_{d_1=m}^{n-1}\sum_{j=0}^{d_1}\sum_{k=1}^{n-d_1}\binom{n}{d_1}\binom{d_1}{j}\binom{n-d_1}{k}(-1)^{j+k}\exp[-\lambda((n-d_1+j)t_1+k(t-t_1))].
\end{align*}
When $m=0$, the stress does not change after $t_1$. Therefore the hazard rate is unchanged up to $t_2$. So, when $m=0$,
\begin{align*}
  R_{Z_n}(t)=\sum_{i=1}^{n}\binom{n}{i}(-1)^{i+1}[R(t)]^i=\sum_{i=1}^{n}\binom{n}{i}(-1)^{i+1}\exp(-\lambda i t) & ~~~~\text{ for } 0<t<t_2.
\end{align*}
Therefore \begin{align*}
    E[\tau]&=\int_{\boldsymbol{\theta}}E[\tau\ | \ \boldsymbol{\theta}]p(\boldsymbol{\theta})~ d\boldsymbol{\theta}\\
    &=\int_{\boldsymbol{\theta}}\int_{t=0}^{t_1}R_{Z_n}(t)~dt~d\boldsymbol{\theta}+\int_{\boldsymbol{\theta}}\int_{t=t_1}^{t_2}R_{Z_n}(t)~dt~d\boldsymbol{\theta}
\end{align*}

\section{ Proof of the monotonicity Property of Bayes decision function :}\label{appA}
\textbf{Proof of Theorem \ref{the}:}
From equation (\ref{e1}), the likelihood function is given by
\begin{align*}
    L(\lambda\ |\ (w_1,w_2,d_1,d_2),\mathbf{q})\propto \lambda^{d}\phi^{ d_2{\delta}}\exp\left[-\lambda(w_1+\phi^{\delta} w_2)\right].
\end{align*}

\begin{enumerate}[(i)]
    \item 

For fixed $(w_2,d_1,d_2)$ , consider two points $w_1^1$ and $w_1^2$ such that $w_1^1<w_1^2$. Then it is sufficient to prove that $\varphi(w_1^1,w_2,d_1,d_2)>\varphi(w_1^2,w_2,d_1,d_2)$ 
when $h(\lambda)$ is increasing function in $\lambda$. Let
$f_1(\lambda)=\lambda^{d}\exp(-\lambda w_1^1)$, $f_2(\lambda)=\lambda^{d}\exp(-\lambda w_1^2)$. Note that $g_1(\lambda)=h(\lambda)\int_{\phi}\phi^{d_2}\exp(-\lambda\phi w_2) p(\boldsymbol{\theta}) d\phi$ and $g_2(\lambda)=\int_{\phi}\phi^{d_2}\exp(-\lambda\phi w_2)p(\boldsymbol{\theta}) d\phi$ if $\delta=1$ and  $g_1(\lambda)=h(\lambda)\exp(-\lambda w_2)$ and $g_2(\lambda)=\exp(-\lambda w_2) $ if ${\delta}=0$. Thus,
\begin{align*}
    \varphi(w_1^1,w_2,d_1,d_2)=\frac{\int_\lambda f_1(\lambda)g_1(\lambda)d\lambda}{\int_\lambda f_1(\lambda)g_2(\lambda)d\lambda}
\end{align*}
and
\begin{align*}
    \varphi(w_1^2,w_2,d_1,d_2)=\frac{\int_\lambda f_2(\lambda)g_1(\lambda)d\lambda}{\int_\lambda f_2(\lambda)g_2(\lambda)d\lambda}.
\end{align*}
We assume that all integrals are finite. Clearly, $f_2(\lambda)$ and $g_2(\lambda)$ are non negative functions of $\lambda$. Now, $f_1(\lambda)/f_2(\lambda)$ $=\exp[\lambda(w^2_1-w^1_1)]$ and $g_1(\lambda)/g_2(\lambda)=h(\lambda)$ which is increasing function in $\lambda$ for $\lambda>0$. By Theorem 2  in Wijsman\cite{wijsman1985useful}, we get $\varphi(w_1^1,w_2,d_1,d_2)>\varphi(w_1^2,w_2,d_1,d_2)$. Therefore, $\varphi(w_1^1,w_2,d_1,d_2)$ is decreasing in $w_1$ for fixed $(w_2,d_1,d_2)$.

\item For ${\delta}=1$,    
the joint posterior distribution of $\lambda$ and $\phi$ is given by
\begin{align*}
    p(\lambda,\phi\ | \ (w_1,w_2,d_1,d_2))=\frac{\lambda^d\phi^{d_2}\exp[-\lambda(w_1+\phi w_2)]~p_1(\lambda)~ p_2(\phi)}{K(w_1,w_2,d_1,d_2)},
\end{align*}
where
$K(w_1,w_2,d_1,d_2)=\int_\lambda\int_\phi\lambda^d\phi^{d_2}\exp[-\lambda(w_1+\phi w_2)] ~p_1(\lambda) ~p_2(\phi)~d\lambda~ d\phi=E_{(\lambda,\phi)}[\lambda^d\phi^{d_2}\exp[-\lambda(w_1+\phi w_2)]]$.
The marginal posterior distribution of $\phi$ is 
\begin{align*}
    p(\phi\ | \ (w_1,w_2,d_1,d_2))=\frac{\int_\lambda\lambda^d\phi^{d_2}\exp[-\lambda(w_1+\phi w_2)]~p_1(\lambda)~ p_2(\phi)~d\lambda}{K(w_1,w_2,d_1,d_2}.
\end{align*}
Therefore, the posterior conditional distribution of  $\lambda | \ \phi$ is given by
\begin{align*}
    p(\lambda\ | \ \phi , (w_1,w_2,d_1,d_2))=\frac{\lambda^d\phi^{d_2}\exp[-\lambda(w_1+\phi w_2)]~p_1(\lambda)~ p_2(\phi)}{\int_\lambda\lambda^d\phi^{d_2}\exp[-\lambda(w_1+\phi w_2)]~p_1(\lambda)~ p_2(\phi)~d\lambda}.
\end{align*}
The posterior expectation of $h(\lambda)$ can be written as
\begin{align*}
   \varphi(w_1,w_2,d_1,d_2)=\int_\phi \int_\lambda h(\lambda)~p(\lambda \ | \ \phi, (w_1,w_2,d_1,d_2))~d\lambda~p(\phi\ | \ (w_1,w_2,d_1,d_2))~d\phi.
\end{align*}
To prove the decreasing property of $\varphi(w_1,w_2,d_1,d_2)$ with respect to $w_2$ for fixed $(w_1,d_1,d_2)$, we need the following two lemmas:
\begin{lemma}\label{l1}
For fixed $(d_1,d_2,w_2)$, $w_2^1$ and $w_2^2$ with $w_2^1<w_2^2$, the likelihood ratio
$${\frac{p(\phi \ | \ (w_1,w_2^1,d_1,d_2)}{p(\phi \ | \ (w_1,w_2^2,d_1,d_2)}}$$  is increasing in $\phi$.
\end{lemma}
\begin{proof}
    \begin{align*}
        \frac{p(\phi \ | \ (w_1,w_2^1,d_1,d_2)}{p(\phi \ | \ (w_1,w_2^2,d_1,d_2)}&=\frac{K(w_1,w_2^2,d_1,d_2)}{K(w_1,w_2^1,d_1,d_2)}\frac{\int_{\lambda} \lambda^d\phi^{d_2}\exp[-\lambda(w_1+\phi w_2^1)]p(\lambda,\phi)d \lambda}{\int_{\lambda} \lambda^d \phi^{d_2}\exp[-\lambda(w_1+\phi w_2^2)]p(\lambda,\phi)d \lambda}\\
        &=\frac{K(w_1,w_2^2,d_1,d_2)}{K(w_1,w_2^1,d_1,d_2)}\frac{\int_{\lambda} \lambda^d\exp[-\lambda(w_1+\phi w_2^1)]p_1(\lambda)d \lambda}{\int_{\lambda} \lambda^d \exp[-\lambda(w_1+\phi w_2^2)]p_1(\lambda) d \lambda}.
    \end{align*}
    Since $\lambda, \phi$ are independent,  $p(\lambda,\phi)=p_1(\lambda)p_2(\phi)$. Since $\lambda^d\phi^{d_2}\exp[-\lambda(w_1+\phi w_2)]\geq0$, $K(w_1,w_2,d_1,d_2)=E_{(\lambda,\phi)}[\lambda^d\phi^{d_2}\exp[-\lambda(w_1+\phi w_2)]\geq 0$. 
    Let $h_1(\phi,w_2,\lambda)=\exp[-\lambda\phi w_2]$, and  $h_2(\lambda)=\lambda^d\exp[-\lambda w_1]p_1(\lambda)$. Now, $h_1'(\phi,w_2,\lambda)=\partial h_1(\phi,w_2,\lambda)/\partial \phi=-w_2\lambda h_1(\phi,w_2,\lambda)$. It is enough to prove that 
\begin{align*}
    h_3(\phi,w_2^1,w_2^2)=\frac{\int_\lambda h_1(\phi,w_2^1,\lambda)h_2(\lambda)d\lambda}{\int_\lambda h_1(\phi,w_2^2,\lambda)h_2(\lambda)d\lambda}
\end{align*}
is increasing in $\phi$.
Now,
\begin{align*}
    \frac{\partial h_3(\phi,w_2^1,w_2^2) }{\partial \phi}=\frac{-(w_2^1-w_2^2)\int_\lambda h_1(\phi,w_2^2,\lambda)h_2(\lambda)d\lambda\int_\lambda w_2^1 \lambda h_1(\phi,w_2^1,\lambda)h_2(\lambda)d\lambda}{(\int_\lambda h_1(\phi,w_2^2,\lambda)h_2(\lambda)d\lambda)^2}\geq 0~~~[\text{as } w_2^1<w_2^2 ]
\end{align*}
This proves the lemma.
\end{proof}
\begin{lemma}\label{lem2} If $h(\lambda)$ is increasing in $\lambda$, then  $ \varphi(w_1,w_2,d_1,d_2,\phi)$
\begin{enumerate}[(i)]
    \item is decreasing in $w_2$ for fixed $(w_1,d_1,d_2,\phi)$.
    \item is decreasing in $\phi$ for fixed $(w_1,w_2,d_1,d_2)$.
\end{enumerate}

\end{lemma} 
\begin{proof}
Note that \begin{align*}
 \varphi(w_1,w_2,d_1,d_2,\phi)  =\int_\lambda h(\lambda)p(\lambda \ | \phi, (w_1,w_2,d_1,d_2)) d\lambda=\frac{\int_0^\infty h(\lambda)\lambda^{d_1+d_2}\phi^{d_2}\exp(-\lambda(w_1+\phi w_2))p({\lambda})p_2(\phi)d\lambda}{\int_0^\infty \lambda^{d_1+d_2}\phi^{d_2}\exp(-\lambda(w_1+\phi w_2))p({\lambda)p_2(\phi)}d\lambda}
\end{align*}
Consider two points $w_2^1$ and $w_2^1$ such that $w_2^1<w_2^2$. Let $g_1(\lambda)=h(\lambda)\lambda^{d_1+d_2}\phi^{d_2}p_1(\lambda)$, $g_2(\lambda)=\lambda^{d_1+d_2}\phi^{d_2}p_1(\lambda)$, $f_1(\lambda)=\exp(-\lambda(w_1+\phi w_2^1))$ and 
$f_2(\lambda)=\exp(-\lambda(w_1+\phi w_2^2))$. Clearly, $f_2(\lambda),g_2(\lambda)$ are non negative functions. Note that ${f_1(\lambda)}/{f_2(\lambda)}=\exp[\lambda\varphi(w_2^2-w_2^1)]$ is increasing in $\lambda$. and $g_1(\lambda)/g_2(\lambda)=h(\lambda)$ which is also increasing  function in $\lambda$ for $\lambda>0$. By Theorem 2  in Wijsman\cite{wijsman1985useful}, we get $\varphi(w_1,w_2^1,d_1,d_2,\phi)>\varphi(w_1,w_2^2,d_1^2,d_2,\phi)$. Therefore, $\varphi(w_1,w_2,d_1,d_2,\phi)$ is decreasing in $w_2$ for fixed $(w_1,d_1,d_2,\phi)$.

Similarly, let $\phi_1<\phi_2$, $g_1(\lambda)=h(\lambda)\lambda^dp(\lambda)$, $g_2(\lambda)=\lambda^dp(\lambda)$, $f_1(\lambda)=\phi_1^{d_2}\exp[-\lambda(w_1+\phi_1 w_2)]p(\phi_1)$ and $f_2(\lambda)=\phi_2^{d_2}\exp[-\lambda(w_1+\phi_2 w_2)]p_2(\phi)$ Then $g_(\lambda)/g_2(\lambda)=h(\lambda)$ is increasing in $\lambda$ and $f_1(\lambda)/f_2(\lambda)=k(\phi_1,\phi_2)\exp[\lambda w_2(\phi_2-\phi_1)]$ is also increasing in $\lambda$. By Theorem 2  in Wijsman\cite{wijsman1985useful}, we get $\varphi(w_1,w_2,d_1,d_2,\phi_1)> \varphi(w_1,w_2,d_1^2,d_2,\phi_2)$. Therefore, $\varphi(w_1,w_2,d_1,d_2,\phi)$ is decreasing in $\phi$ for fixed $(w_1,w_2,d_1,d_2)$.    
\end{proof}
\vspace{0.2cm}\\
\noindent Now, we prove that $\varphi(w_1,w_2^1,d_1,d_2)-\varphi(w_1,w_2^2,d_1,d_2)\geq  0$ for $w_2^1<w_2^2$ and $h(\lambda)$ is increasing in $\lambda$. The proof is in a similar manner as given in Lemma 3.4.2 of the book of Lehmann \& Romano \cite{lehmann1986testing}.
\begin{flalign}\label{a1}
   &~~~~~~ \varphi(w_1,w_2^1,d_1,d_2)-\varphi(w_1,w_2^2,d_1,d_2)\nonumber\\&=\int_{\phi}\varphi(w_1,w_2^1,d_1,d_2,\phi)p(\phi \ | \ (w_1,w_2^1,d_1,d_2)d\phi-\int_{\phi}\varphi(w_1,w_2^2,d_1,d_2,\phi)p(\phi \ | \ (w_1,w_2^2,d_1,d_2)d\phi\nonumber\\
    &\geq \int_{\phi}\varphi(w_1,w_2^2,d_1,d_2,\phi)p(\phi \ | \ (w_1,w_2^1,d_1,d_2)d\phi-\int_{\phi}\varphi(w_1,w_2^2,d_1,d_2,\phi)p(\phi \ | \ (w_1,w_2^2,d_1,d_2)d\phi\nonumber\\
    &~~~~~~~~~~~~~~~~~~~~~~~~~~~~~~~~~~~~~~~~~~~~~~~~~~~~~~~~~~~~~~~~~~~~~~~~~~~~~~~~~~~~~[\text{using Lemma \ref{lem2}(i)}]\nonumber\\
    &=\int_\phi \varphi(w_1,w_2^2,d_1,d_2,\phi)[p(\phi \ | \ (w_1,w_2^1,d_1,d_2))-p(\phi \ | \ (w_1,w_2^2,d_1,d_2)] ~d\phi.
\end{flalign}
Define the sets U and V as  $U=\{\phi \ |\  p(\phi \ |\ (w_1,w_2^1,d_1,d_2))> p(\phi \ |\ (w_1,w_2^2,d_1,d_2))\}$ and $V=\{\phi \ |\  p(\phi  | (w_1,w_2^1,d_1,d_2)$ $ \leq p(\phi \ |\ (w_1,w_2^2,d_1,d_2))\}$ and let $u=\inf\limits_U\{\varphi(w_1,w_2^2,d_1,d_2,\phi)\}$ and $v=\sup\limits_V\{\varphi(w_1,w_2^2,d_1,d_2,\phi)\}$.
From equation (\ref{a1}), we have
\begin{align}\label{l2}
    &\varphi(w_1,w_2^1,d_1,d_2)-\varphi(w_1,w_2^2,d_1,d_2)\nonumber\\&\geq\int\limits_U \varphi(w_1,w_2^2,d_1,d_2,\phi)[p(\phi \ | \ (w_1,w_2^1,d_1,d_2))-p(\phi \ | \ (w_1,w_2^2,d_1,d_2)] ~d\phi\nonumber\\
    &~~~~~~~~~~~~~~~~~~~~~+\int\limits_V \varphi(w_1,w_2^2,d_1,d_2,\phi)[p(\phi \ | \ (w_1,w_2^1,d_1,d_2))-p(\phi \ | \ (w_1,w_2^2,d_1,d_2)] ~d\phi\nonumber\\
    &\geq u\int\limits_U [p(\phi \ | \ (w_1,w_2^1,d_1,d_2))-p(\phi \ | \ (w_1,w_2^2,d_1,d_2)] ~d\phi+ v\int\limits_V [p(\phi \ | \ (w_1,w_2^1,d_1,d_2))-p(\phi \ | \ (w_1,w_2^2,d_1,d_2)] ~d\phi\nonumber\\
    &=(u-v)\int\limits_U [p(\phi \ | \ (w_1,w_2^1,d_1,d_2))-p(\phi \ | \ (w_1,w_2^2,d_1,d_2)] ~d\phi
\end{align}
The sets $U$ and $V$ can be written as

$$U=\left\{\phi\ | \ \frac{p(\phi \ |\ (w_1,w_2^1,d_1,d_2))}{p(\phi \ |\ (w_1,w_2^2,d_1,d_2))}>1\right\}$$
and $$V=\left\{\phi\ | \ \frac{p(\phi \ |\ (w_1,w_2^1,d_1,d_2)) }{ p(\phi \ |\ (w_1,w_2^2,d_1,d_2))}<1\right\}$$.
For all $\phi_1\in U$ and $\phi_2 \in V$, we have
\begin{align*}
    \frac{p(\phi_1 \ |\ (w_1,w_2^1,d_1,d_2))}{p(\phi_1 \ |\ (w_1,w_2^2,d_1,d_2))}>1> \frac{p(\phi_2 \ |\ (w_1,w_2^1,d_1,d_2))}{p(\phi_2 \ |\ (w_1,w_2^2,d_1,d_2))}.
\end{align*}
From Lemma \ref{l1}, $${\frac{p(\phi \ | \ (w_1,w_2^1,d_1,d_2)}{p(\phi \ | \ (w_1,w_2^2,d_1,d_2)}}$$  is increasing in $\phi$. Therefore
\begin{align*}
 \frac{p(\phi_1 \ |\ (w_1,w_2^1,d_1,d_2))}{p(\phi_1 \ |\ (w_1,w_2^2,d_1,d_2))}> \frac{p(\phi_2 \ |\ (w_1,w_2^1,d_1,d_2))}{p(\phi_2 \ |\ (w_1,w_2^2,d_1,d_2))}\implies \phi_1>\phi_2.
\end{align*}
Now, from Lemma \ref{lem2} (ii), $\varphi(w_1,w_2,d_1,d_2,\phi)$ is a decreasing function in $\phi$. Therefore we get $\varphi(w_1,w_2,d_1,d_2,\phi_1)<\varphi(w_1,w_2,d_1,d_2,\phi_2)$, $\forall ~\phi_1\in U$ and $\phi_2\in V$. This implies that $\inf\limits_U\{\varphi(w_1,w_2,d_1,d_2,\phi)\}<\sup\limits_V\{\varphi(w_1,w_2,d_1,d_2,\phi)\}=u>v$. From equation  (\ref{l2}) we get $\varphi(w_1,w_2^1,d_1,d_2)-\varphi(w_1,w_2^2,d_1,d_2)\geq0$. Therefore $\varphi(w_1,w_2,d_1,d_2)$ is decreasing in $w_2$ for fixed $(w_1,d_1,d_2)$
\end{enumerate}

\section{The joint distribution of \texorpdfstring{$(W_1,W_2,D_1,D_2)$}{}}\label{appb2}
\begin{theorem}(Balakrishnan et. al \cite{balakrishnan2009exact})
The joint distribution of $(W_1,W_2,D_1,D_2\  | \  \lambda,\phi)$ is given below:
\begin{align*}
&f_{(W_1,W_2,D_1,D_2)\ | \ \lambda,\phi)}(w_1,w_2,d_1,d_2\ | \ \lambda,\phi)\\
=&\sum_{j=0}^{d_1}\sum_{k=0}^{d_2}A_{d_1d_2jk}p_1(\lambda)p_2(\lambda,\phi)\begin{dcases*}  
I(nt_1)I(n(t_2-t_1))& $d_1=0,d_2=0$\\
I(nt_1)~ \gamma\left(w_2-T_{d_1d_2k},d_2,{\phi^{\delta}\lambda}\right)&$d_1=0,d_2>0$\\
\gamma\left(w_1-T_{d_1j},d_1,{\lambda}\right)I(n(t_2-t_1))& $d_1>0,d_2=0$\\
\gamma\left(w_1-T_{d_1j},d_1,{\lambda}\right)\gamma\left(w_2-T_{d_1d_2k},d_2,{\phi^{\delta}\lambda}\right)&$d_1>0,d_2>0$,
\end{dcases*}
\end{align*}
where 
  $A_{d_1d_2jk}=\frac{n!}{d_1!d_2!(n-d)!}\binom{d_1}{j}\binom{d_2}{k}(-1)^{(j+k)}$, $p_1(\lambda)=\exp(-\lambda(n-d_1+j)t_1)$ and $p_2(\lambda,\phi)=\exp(-\phi^{\delta}\lambda(n-d+k)(t_2-t_1))$,
    the function $\gamma(x,\alpha,\beta)$ is the PDF of the gamma distribution with parameters $(\alpha,\beta)$ which is defined as
\begin{align*}  \gamma(x,\alpha,\beta)=\begin{cases}
    \frac{\beta^{\alpha}}{\Gamma(\alpha)}x^{\alpha-1}\exp(-\beta x)& \text{for } x>0\\
    0 & \text{otherwise},
    \end{cases}
\end{align*}
and  the function $I(x)$ is a degenerate distribution at the point $a$ which is defined as \begin{align}
    I(x)=\begin{cases}
    1& \text{for } x=nT_0\\
    0 & \text{otherwise},
    \end{cases}
\end{align}
\end{theorem}

\section{Expression of Bayes risk for quadratic loss function}\label{appb1}
Consider the quadratic loss function $h(\lambda)=a_0+a_1\lambda+a_2\lambda^2$. The joint distribution of $\lambda$ and $\phi$ is given by
\begin{align*}
    p(\lambda,\phi)=\frac{\beta^\alpha}{\Gamma(\alpha)(l-1)}\lambda^{\alpha-1}\exp(-\beta\lambda).
\end{align*}
The Bayes risk is then obtained as
\begin{align*}
    R_B(\boldsymbol{q},a)=n(C_s-v_s)+C_a n_{as}+E[D]+C_tE[\tau]+E_{\lambda}[h(\lambda)]+\sum_{d_1=0}^n\sum_{d_2=0}^{n-d_1}H(d_1,d_2),
\end{align*}
where $E_\lambda[h(\lambda]$, $H(d_1,d_2)$, $E[D]$, $n_{as}$ and $E[\tau]$ are as follows:
$$E_{\lambda}[h(\lambda)]=a_0+a_1\frac{\alpha}{\lambda}+a_2\frac{\alpha(\alpha+1)}{\beta^2}$$,
\scriptsize
\begin{align*}
    H(d_1,d_2)=\begin{dcases*}
        \frac{\beta^\alpha}{\Gamma(\alpha)(l-1)}\sum_{p=0}^2a_p~I(nt_1\leq c(0,0,n(t_2-t_1))~\&~n(t_2-t_1)\leq c(0,0))H_1(0,0,nt_1,n(t_2-t_1))& for $d_1=0,d_2=0$\\
        \frac{\beta^\alpha}{\Gamma(\alpha)(l-1)}\sum_{k=0}^{d_2}\int_{w_2=T_{0,d_2,k}}^{\xi(0,d_2,k)}~\sum_{p=0}^2a_pI(nt_1\leq c(0,d_2,w_2))\frac{(w_2-T_{0,d_2,k})^{d_2-1}}{\Gamma(d_2)} ~H_1(nt_1,w_2,0,d_2,p)~dw_2 & for $d_1=0,d_2>0$\\
\frac{\beta^\alpha}{\Gamma(\alpha)(l-1)}\sum_{j=0}^{d_1}\int\limits_{w_1=T_{d_1,j}}^{\xi(d_1,0,k)}\sum_{p=0}^2a_p~I((n-d_1)(t_2-t_1)\leq c(d_1,0))\frac{(w_1-T_{d_1j})^{d_1-1}}{\Gamma(d_1)}H_1(w_1,n(t_2-t_1),d_1,0,p)~ dw_1 & for $d_1>0,d_2=0$\\
        \frac{\beta^\alpha}{\Gamma(\alpha)(l-1)}\sum_{j=0}^{d_1}\sum_{k=0}^{d_2}\int\limits_{w_2=T_{d_1,d_2,k}}^{\xi(d_1,d_2,k)}\int\limits_{w_1=T_{d_1,j}}^{\xi(d_1,d_2,w_2,j)}\sum_{p=0}^2a_p\frac{(w_1-T_{d_1,j})^{d_1-1}(w_2-T_{d_1,d_2,k})^{d_2-1}}{\Gamma(d_1)\Gamma(d_2)}H_1(w_1,w_2,d_1,d_2,p)~dw_1~dw_2  & for $d_1>0,d_2>0$.
    \end{dcases*}
\end{align*}
\begin{align*}
    E[D]=\begin{dcases*}
        n-n\left(\frac{\beta}{\beta+t_2}\right)^\alpha &for $m=0$\\
        n+\sum_{d_1=0}^{m-1}(n-d_1)\binom{n}{d_1}\sum_{i=0}^{d_1}\binom{d_1}{i}(-1)^{d_1-i+1}\left[\frac{1}{(l-1)(t_2-t_1)(\alpha-1)}\left\{\frac{\beta^\alpha}{[\beta+(n-i)t_1+(t_2-t_1)]^{\alpha-1}}\right.\right.\\
      \left.\left. -\frac{\beta^\alpha}{[\beta+(n-i)t_1+l(t_2-t_1)]^{\alpha-1}}\right\}\right]  +\sum_{d_1=m}^{n-1}(n-d_1)\binom{n}{d_1}\sum_{i=0}^{d_1}\binom{d_1}{i}(-1)^{d_1-i+1}\left[\left(\frac{\beta}{\beta+(n-i)t_1+(t_2-t_1)}\right)^\alpha\right], &for  $1\leq m\leq n-1$\\
    n-\sum_{d_1=0}^{n-1}(n-d_1)\binom{n}{d_1}\sum_{i=0}^{d_1}\binom{d_1}{i}\left[\frac{(-1)^{d_1-i+1}}{(l-1)(t_2-t_1)(\alpha-1)}\left\{\frac{\beta^\alpha}{[\beta+(n-i)t_1+(t_2-t_1)]^{\alpha-1}}-\frac{\beta^\alpha}{[\beta+(n-i)t_1+l(t_2-t_1)]^{\alpha-1}}\right\}\right],~~ &for   $m=n$,
    \end{dcases*}
\end{align*}
\normalsize
\begin{align*}
 n_{as}=\begin{dcases*}
     0& for $m=0$\\
      \sum_{d_1=0}^{m-1}\sum_{i=0}^{d_1}(n-d_1)\binom{n}{d_1}\binom{d_1}{i}(-1)^i\left(\frac{\beta}{(n-i)t_1+\beta}\right)^{\alpha} &for $m> 0$ 
 \end{dcases*}   
\end{align*}
   \begin{align*}
       E[\tau]=\begin{dcases*}
           \sum_{i=1}^{n}\binom{n}{i}(-1)^{i}\left[\frac{\beta}{i(\alpha-1)}-\frac{\beta^\alpha}{i(\alpha-1)(\beta+it_2)^{\alpha-1}}\right] &for $m=0$\\
          \int_{\boldsymbol{\theta}}\int_{t=0}^{t_1}R_{Z_n}(t)~dt~d\boldsymbol{\theta} +   \int_{\boldsymbol{\theta}}\int_{t=0}^{t_2}R_{Z_n}(t)~dt~d\boldsymbol{\theta} & for $m> 0$,
       \end{dcases*}
   \end{align*}
   where
\begin{align*}
  \int_{\boldsymbol{\theta}}\int_{t=0}^{t_1}R_{Z_n}(t)~dt~d\boldsymbol{\theta} =&\int_0^\infty\int_0^{t_1} \sum_{i=1}^{n}\binom{n}{i}(-1)^{i+1}\exp(-i\lambda t)\frac{\beta^\alpha}{\Gamma(\alpha)}\lambda^{\alpha-1}\exp(-\beta\lambda)~dt~d\lambda\\
    =&\frac{\beta^\alpha}{\Gamma(\alpha)}\sum_{i=1}^{n}\binom{n}{i}(-1)^{i+1}\int_0^\infty\lambda^{\alpha-1}\exp(-\beta\lambda)\left[\left(\frac{1-\exp(-i\lambda t_1)}{i\lambda}\right)\right]~d\lambda\\
    =&\sum_{i=1}^{n}\binom{n}{i}(-1)^{i}\left[\frac{\beta}{i(\alpha-1)}-\frac{\beta^\alpha}{i(\alpha-1)(\beta+it_1)^{\alpha-1}}\right]
\end{align*}
and
\begin{align*}
  &\int_{\boldsymbol{\theta}}\int_{t=0}^{t_2}R_{Z_n}(t)~dt~d\boldsymbol{\theta}\\
  =& \sum_{d_1=0}^{n-1}\sum_{j=0}^{d_1}\sum_{k=1}^{n-d_1}\frac{A_{d_1n-d_1jk}}{m-1}\int_1^l\left[\frac{\beta^\alpha}{(\beta+(n-d_1+j)t_1)^\alpha\phi^{\delta}(\alpha-1)}-\frac{\beta^\alpha}{(\beta+(n-d_1+j)t_1+\phi^{\delta} k(t_2-t_1))^\alpha\phi^{\delta}(\alpha-1)}\right]~d\phi
\end{align*}

\end{document}